\documentclass[11pt, a4paper]{googledeepmind}

\usepackage{microtype}
\usepackage{graphicx}
\usepackage{subcaption}
\usepackage{multirow}
\usepackage{booktabs}
\usepackage{natbib}

\setcitestyle{numbers}
\setcitestyle{square}
\setcitestyle{comma}

\usepackage{setspace}
\usepackage[dvipsnames]{xcolor}
\usepackage{longtable}
\usepackage{array}
\usepackage{tabularx}
\usepackage{wrapfig}
\usepackage{bbm}
\usepackage{amsmath}
\usepackage{amssymb}
\usepackage{mathtools}
\usepackage{amsthm}
\usepackage{algorithm}
\usepackage{algorithmic}
\usepackage{amsfonts}
\usepackage[utf8]{inputenc}
\usepackage[T1]{fontenc}
\usepackage{hyperref}
\definecolor{citeColor}{RGB}{0,20,115}
\hypersetup{colorlinks,linkcolor={citeColor},citecolor={citeColor},urlcolor={citeColor}}
\usepackage{url}
\usepackage{nicefrac}
\usepackage{enumerate}
\usepackage{bm}
\usepackage{enumitem}
\usepackage{makecell}
\usepackage{colortbl}
\usepackage{bbding}
\usepackage{caption}
\usepackage{multibib}
\usepackage{tcolorbox}
\usepackage{fontawesome6}
\tcbuselibrary{skins,breakable}

\newtcolorbox{promptbox}[1]{
  enhanced, breakable,
  colback=gray!5, colframe=black!70,
  colbacktitle=black!75, coltitle=white,
  fonttitle=\small\bfseries,
  title={#1},
  boxrule=0.5pt, arc=1.5pt,
  left=6pt, right=6pt, top=4pt, bottom=4pt,
  fontupper=\small,
}

\newtcolorbox{prompttemplate}[1]{
  enhanced, breakable,
  colback=gray!5, colframe=black!60,
  colbacktitle=black!70, coltitle=white,
  fonttitle=\small\bfseries,
  title={#1},
  boxrule=0.5pt, arc=1.5pt,
  left=6pt, right=6pt, top=4pt, bottom=4pt,
  fontupper=\small\ttfamily,
}

\usepackage[normalem]{ulem}
\usepackage{varwidth}
\usepackage{tikz}
\usetikzlibrary{calc}
\usetikzlibrary{positioning, shapes, fit, backgrounds, decorations.pathreplacing}
\usetikzlibrary{arrows.meta}

\usepackage[capitalize,noabbrev]{cleveref}

\usepackage{etoc}
\etocdepthtag.toc{mtchapter}
\etocsettagdepth{mtchapter}{subsection}
\etocsettagdepth{mtappendix}{none}

\theoremstyle{plain}
\newtheorem{theorem}{Theorem}[section]
\newtheorem{proposition}[theorem]{Proposition}
\newtheorem{lemma}[theorem]{Lemma}
\newtheorem{corollary}[theorem]{Corollary}
\theoremstyle{definition}

\newtheorem{assumption}[theorem]{Assumption}

\theoremstyle{remark}
\newtheorem{remark}[theorem]{Remark}

\usepackage{pifont}

\definecolor{boxfill}{RGB}{246,243,235}
\definecolor{nullfill}{RGB}{235,230,216}
\definecolor{boxedge}{RGB}{130,124,108}
\definecolor{inkdark}{RGB}{28,27,24}
\definecolor{inkdim}{RGB}{90,88,82}
\definecolor{accent}{RGB}{170,100,40}

\newcommand{\Excl}{\mathrm{Excl}}
\newcommand{\HHI}{\mathrm{HHI}}
\newenvironment{restatement}[1]{\par\medskip\noindent\textbf{#1.}\itshape\enskip}{\par\medskip}

\title{When AI Writes, Who Gets Cited? Evidence of Citation Monoculture Across Language Models}

\author{
    \textbf{Sina Alemohammad}$^{1,\dagger}$\;
    \textbf{Denghui Zhang}$^{2}$\;
    \textbf{Bolong Tang}$^{1}$\;
    \textbf{Anthony Qin}$^{3}$\;
    \textbf{Gengchen Mai}$^{1}$
    \\
    \textbf{Ahmed Abbasi}$^{4}$\;
    \textbf{Richard Baraniuk}$^{5}$\;
    \textbf{Zhangyang Wang}$^{1}$
    \vspace{2mm} \\
    $^{1}$The University of Texas at Austin \quad
    $^{2}$Stevens Institute of Technology \\
    $^{3}$Washington University in St.\ Louis \quad
    $^{4}$University of Notre Dame \\
    $^{5}$Rice University \\[2mm]
    \textbf{\textcolor{MidnightBlue}{Code and data}:
    \url{https://github.com/VITA-Group/Citation-Collapse}}
    \\[-6mm]
}

\correspondingauthor{$^\dagger$Correspondence to sinaalemohammad@gmail.com.}

\begin{abstract}
As language models move from drafting prose to running literature-search agents with tool calls, fabricated references are becoming easier to catch and constrain. The harder failure begins after every candidate is real: different models may still select the same narrow subset, producing \textbf{citation monoculture} without any single citation being wrong. We isolate this effect on 120 real papers. Eleven models from three vendors choose at most ten papers from uniformly random panels of thirty, with real titles and abstracts but fabricated authors, reassigned years, and hidden venues and citation counts. Each run is compared with indifferent selection on the same panel and realized budget. We find all eleven models to concentrate sharply. The top decile receives 23.3--30.2\% of citations against 15.6\% under the null, and one component explains 68--73\% of variation across their preference maps; cross-vendor agreement nearly matches within-vendor agreement. We mathematically formalize the task as fixed-budget subset selection. Our theory turns these aggregate patterns into identifiable mechanisms: an exchangeability bound shows that the observed exclusions require stable paper-level favoritism; a spectral decomposition explains how shared preference produces concentration and limits model mixing; and a rarity theorem predicts the recursive competition effect verified within panels. Each theoretical result connects directly to the experiments: the bound rejects a mapless selector for every model, the decomposition explains why the best cross-fitted mixture still retains 55\% of the excess, and the rarity theorem is confirmed by a within-panel competition test. Controlled paraphrase, content-slot crossover, and design resampling attribute about 90\% of GPT-5 mini's map variance to paper content, with the dominant content effect replicated in GPT-4.1 mini. Eight domain experts select from the same blinded panels under the same cap; their pooled choices show no comparable shared preference, while model concentration persists in selection-only mode. Thus, \textit{even when every reference is real and every paper is equally visible, current language models impose a common content-level filter on scientific attention. Equalizing retrieval or mixing vendors is therefore insufficient; the shared preference map itself must be changed.}
\end{abstract}

\begin{document}

\maketitle

\begin{figure}[!ht]
    \centering
    \includegraphics[width=\textwidth]{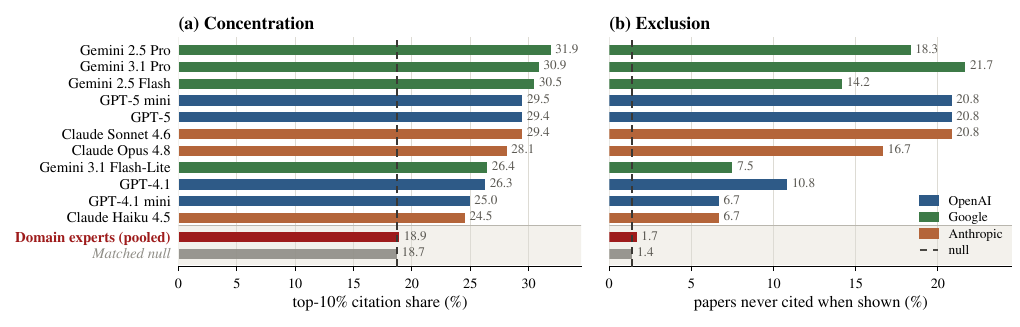}
    \caption{\textbf{Every model concentrates citations and excludes
    papers far beyond indifferent selection; domain experts show no comparable
    effect.} On the
    53 prompts an expert also completed, with byte-identical panels and
    budgets: (a)~top-10\% citation share and (b)~shown papers never cited
    across roughly thirteen exposures. Every model sits well right of the
    matched null, the same model choosing indifferently on the same panel
    (dashed: 18.7\% and 1.4\%); pooled experts sit at the null on both
    margins (18.9\% and 1.7\%). Every model also falls outside the entire simulated null support on the
    exchangeability test,
    while the experts sit inside it.}
    \label{fig:teaser}
\end{figure}


\section{Introduction}
\label{sec:intro}
 
For fifty years the concentration of scientific credit has had a name, the
Matthew effect, and a known engine made of status, visibility, and
accumulated citation count \citep{merton1968matthew, price1976general}. That
engine has always run through human readers, who can see who is already rich
and cite accordingly. The reader is now changing. Generative models are
moving from polishing sentences to drafting the literature-review and
citation-selection steps that shape which work receives attention
\citep{liang2024mapping, hao2024expand}, and the rich-get-richer dynamic has
to be asked again for a citer that has never been shown who is rich. What
gets audited today is fabrication, references that do not exist
\citep{walters2023fabrication, agrawal2024language}, but a system can produce
locally valid references while narrowing scientific attention globally. We
ask whether it does. When every candidate is real, directly shown, and
stripped of every displayed signal of status or count, does citation credit
still condense, and is the condensation one shared thing?

Measuring this in the wild is hopeless, because a model's taste is confounded
with retrieval rankings, prestige cues, and fabrication. Our benchmark
removes them one at a time. A fixed corpus of 120 real papers, from which
every prompt draws a uniformly random panel of 30, removes retrieval.
Candidates keep their real titles and abstracts, but their author names are
fabricated, their years randomly reassigned, and their citation counts
hidden, which removes prestige. Hallucinated citations are rare and are
removed before analysis. Each prompt may cite at most 10 of the 30 shown,
because without scarcity models cite nearly everything and look indifferent.
The budget is what forces a choice. Every run is then compared against the
same model, on the same panel, under the same budget, choosing indifferently,
so that every excess reported here means distance from a model's own behavior
under indifference.

What survives the removals is large and, more importantly, singular. Every
model tested concentrates citations well beyond its null and persistently
leaves papers uncited despite dozens of exposures, and concentration,
exclusion, and the homogenization of bibliographies all move together, a
triplet we call citation collapse (Figure~\ref{fig:teaser}). Behind it lies
nearly one object, a single preference map, a stable ranking of the same
papers applied regardless of prompt, which eleven models trained by three
independent organizations reproduce on blinded abstracts, with agreement
across vendors nearly matching agreement within them. This is algorithmic
monoculture \citep{kleinberg2021algorithmic, bommasani2022picking}, measured
rather than posited, and it prices the obvious remedy. Because the map is
shared, mixing models removes only its idiosyncratic part, and the best
mixture over the eleven still leaves more than half the excess. Whether that
map reflects an absorbed citation consensus or a convergent judgment of what
a citable paper reads like, every result below holds under either reading.

The natural excuse is relevance, that any competent reader of the same
abstracts would reproduce the shared ranking. So we gave the identical task
to domain experts. Eight of them completed 53 of the 120 prompts, 41 by a
single annotator, with the same panels, the same blinding, and the same
budget, and their pooled selections are statistically indistinguishable from
indifferent choice. A consequence of exchangeability bounds how often a shown
paper can go entirely uncited, and every model falls outside the entire
simulated null support while the experts sit inside it. Given the size
and imbalance of the panel we read this as a bound on how much map the
experts could be carrying undetected. A protocol-matched sensitivity
calculation puts it at 7.5\% of the average model's, rather than as a
difference in kind.

One-shot use is not how this technology will be used, so we run the benchmark
recursively, round after round, with model-written papers joining the
candidate pool until the original papers are rare in any panel. Exposure
stays random, and being cited never changes what is shown next, so this is
not a citation-to-exposure feedback loop. What changes is the competition. Each surviving real
paper is cited on a rising share of its appearances, while the real papers'
share of all citations falls from everything to 15.6\%. A fixed preference
meeting a thinning class predicts the first of these, but the trajectory
cannot separate rarity from the other things that shift alongside it, so we
test the mechanism inside single panels at a fixed round, where no drift over
time can contribute. The filter does not fade with repeated use, and it is
not amplified by it either. It is concentrated onto fewer targets.

One benchmark cannot settle what models do to a literature. Ours is a single
topic, a narrow range of task formats, and a budget scarce enough that preference has
somewhere to show. It measures choice among papers already placed in front of
the model, not retrieval, not what reaches publication, and not what authors
do with a draft they are handed. What it can do is separate a model's
preference from the status signals that normally travel with a citation,
which is the one thing observational data cannot.

This paper makes four contributions.
\begin{itemize}
    \item \textbf{A benchmark} that isolates citation choice from retrieval,
    prestige, and fabrication: real papers, random exposure, blinded
    metadata, a hard citation budget, and a null that replays every choice
    indifferently on the same panel.
    \item \textbf{A theory exact for that benchmark}, rather than adapted
    from network growth: the null as its uniform case, an identity equating
    excess concentration with bibliography homogenization, an impossibility
    bound that any selector free of stable paper-level favoritism must
    satisfy, and a rarity theorem for
    recursive use.
    \item \textbf{The measured amplifier}: a single preference map shared
    across models and vendors, not detected in domain experts on identical
    panels, predictable from paper text alone, carried by the record's
    title and abstract rather than by anything else attached to it, and not
    removed by mixing models, whose best available combination still leaves
    more than half the excess.
    \item \textbf{What repeated use does}: a fixed preference meeting a
    diluting catalogue raises the rate at which each surviving real paper is
    cited while lowering the share of citations they receive together, with
    the mechanism verified within single panels at fixed time rather than
    inferred from trajectories.
\end{itemize}


\section{The Game and Its Null}
\label{sec:benchmark}

\textbf{The corpus and the game.}
The benchmark is built on 120 real knowledge-distillation papers collected
from arXiv (published 2015--2022) with between 50 and 500 citations at
collection time, a band that excludes both obscure and canonical work. Every
paper is shown as a blinded record. It keeps its real title and abstract,
but we replace author names with fabricated ones, reassign publication
years, and hide citation counts and venues, so no prestige signal survives
(App.~\ref{app:corpus}). Each round consists of $M = 120$ prompts, half
asking for a short review and half for a position piece. A prompt sees a
\emph{panel} of $m=30$ papers drawn uniformly at random from the current
catalogue and may cite at most 10 of them; citations are parsed
deterministically from the generated text, and hallucinated references, a
mean per-prompt rate of at most 1.2\% for any model, are removed before
every analysis (App.~\ref{app:parsing}). Because each null inherits the
realized budget of its prompt, and removal can lower that budget, the null is conditioned on
a quantity the model's own behavior can move. The dependence is slight, since mean distinct papers cited runs 9.62--10.00
across models, but it is a dependence of the baseline on the outcome and we
note it rather than assume it away. Round 0 is the static benchmark, run on eleven
models from three vendors. In the recursion, each round's 120 model-written
papers join the catalogue of the next round, so the catalogue grows by 120
papers each round, to 1{,}440 at round 11, while panels stay uniformly
random.
Figure~\ref{fig:pipeline} shows the pipeline; Table~\ref{tab:design} lists
the design parameters.

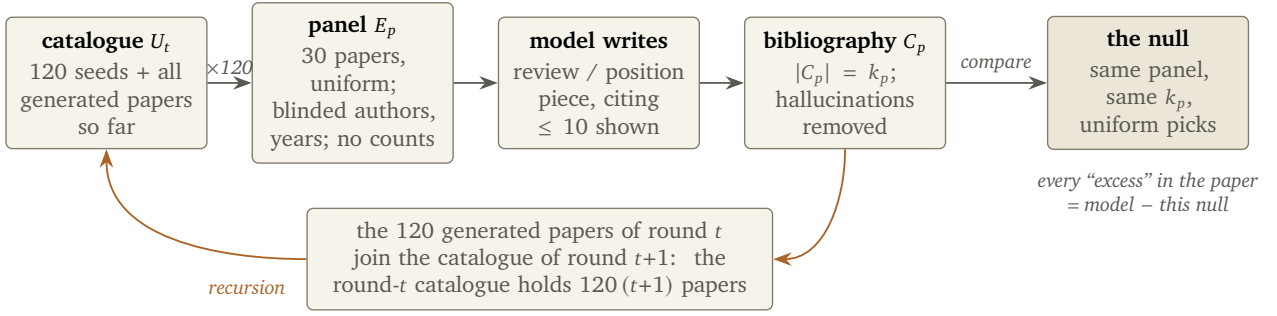
\begin{figure}[t]
        \centering
    \resizebox{\textwidth}{!}{%
    \begin{tikzpicture}[
        font=\footnotesize,
        stage/.style={
            draw=boxedge, line width=0.55pt, rounded corners=2.5pt,
            fill=boxfill, text width=2.62cm, align=center,
            minimum height=1.85cm, inner xsep=3pt, inner ysep=5pt
        },
        flow/.style={-{Stealth[length=2.6mm,width=1.9mm]}, line width=0.7pt, draw=inkdim},
        loop/.style={-{Stealth[length=2.6mm,width=1.9mm]}, line width=0.75pt, draw=accent},
        lab/.style={font=\scriptsize\itshape, text=inkdim, inner sep=1pt},
        node distance=0.62cm
    ]
    \node[stage] (cat) {%
        \textbf{catalogue $U_t$}\\[1.5pt]
        {\color{inkdim}120 seeds $+$ all\\ generated papers\\ so far}};
    \node[stage, right=of cat] (panel) {%
        \textbf{panel $E_p$}\\[1.5pt]
        {\color{inkdim}30 papers, uniform;\\ blinded authors,\\ years; no counts}};
    \node[stage, right=of panel] (write) {%
        \textbf{model writes}\\[1.5pt]
        {\color{inkdim}review / position\\ piece, citing\\ $\le 10$ shown}};
    \node[stage, right=of write] (bib) {%
        \textbf{bibliography $C_p$}\\[1.5pt]
        {\color{inkdim}$|C_p| = k_p$;\\ hallucinations\\ removed}};
    \node[stage, right=1.42cm of bib, fill=nullfill] (null) {%
        \textbf{the null}\\[1.5pt]
        {\color{inkdim}same panel,\\ same $k_p$,\\ uniform picks}};
    \draw[flow] (cat) -- (panel)
        node[lab, midway, above=2pt] {$\times$120};
    \draw[flow] (panel) -- (write);
    \draw[flow] (write) -- (bib);
    \draw[flow] (bib) -- (null)
        node[lab, midway, above=3pt] {compare};
    \node[stage, text width=6.35cm, minimum height=1.02cm]
        (rec) at ($(panel.south)!0.38!(bib.south) + (0,-1.42)$) {%
        {\color{inkdim}the 120 generated papers of round $t$ join the catalogue
        of round $t{+}1$:\; the round-$t$ catalogue holds $120\,(t{+}1)$ papers}};
    \draw[loop] (bib.south) to[out=-90, in=0] (rec.east);
    \draw[loop] (rec.west) to[out=180, in=-90] (cat.south);
    \node[lab, text=accent, anchor=north east] at ($(rec.west)+(-0.25,-0.24)$)
        {recursion};
    \node[lab, anchor=north, align=center] at ($(null.south)+(0,-0.28)$)
        {every ``excess'' in the paper\\ $=$ model $-$ this null};
    \end{tikzpicture}}
    \caption{\textbf{The benchmark pipeline.} A prompt draws a uniformly
    random panel of 30 papers from the current catalogue, shown with blinded
    metadata; the model writes a short review or position piece
    citing at most 10 of them; the parsed bibliography is compared with the
    null: the same panel and the same realized budget, with uniform picks. In
    the recursion, each round's 120 generated papers join the next round's
    catalogue. Every ``excess'' reported in this paper means model minus
    null.}
    \label{fig:pipeline}
\end{figure}

\begin{table}[t]
    \centering
    \caption{Design parameters. Parsing rules, prompt templates, and
    compliance details are in App.~\ref{app:benchmark}.}
    \label{tab:design}
    \footnotesize
    \begin{tabular}{@{}ll@{}}
        \toprule
        Seed corpus & 120 real knowledge-distillation papers (arXiv 2015--2022; 50--500 citations) \\
        Displayed metadata & real title and abstract; fabricated authors; randomly reassigned years; counts hidden \\
        Prompts per round & 120 (60 literature reviews, 60 position pieces) \\
        Panel & $m=30$ papers, uniform without replacement from the catalogue \\
        Budget & at most 10 citations, all from the shown panel \\
        Rounds & 12; the round-$t$ catalogue holds $120\,(t{+}1)$ papers ($1{,}440$ at round 11) \\
        Models & round 0: eleven models, three vendors; recursion: eight (four OpenAI, four Gemini) \\
        Null & same panel, same realized budget, uniform picks (150 redraws at round 0; 40 per round) \\
        \bottomrule
    \end{tabular}
\end{table}

\textbf{The choice law.}
Write $U_t$ for the catalogue at round $t$, of size $N_t$; $E_p \subset U_t$
for the panel shown to prompt $p$; and $C_p \subseteq E_p$ for its parsed
bibliography, with realized budget $k_p = |C_p| \le 10$. We model the
selection as a weighted choice of a $k_p$-subset. Each paper $j$ carries a
latent weight $w_{pj} > 0$ under prompt $p$, and
\begin{equation}
    \Pr\left(C_p = C \mid E_p, k_p\right)
    \;=\;
    \frac{\prod_{j \in C} w_{pj}}{e_{k_p}\!\left(w_{p,\,E_p}\right)},
    \qquad |C| = k_p,
    \label{eq:choicelaw}
\end{equation}
where $e_k$ is the elementary symmetric polynomial over the shown weights.
This is the conditional law of independent inclusion decisions given the
realized budget, so it assumes no coordination across papers beyond the
budget itself. The weights are indexed by prompt, so the law does not
presume a single prompt-invariant ranking; whether one selector's $M$
prompts in fact carry one map is the empirical question
Section~\ref{sec:onemap} answers. Two symbols recur throughout: the probability that prompt $p$ cites a
shown paper, $r_{pj} = \Pr(j \in C_p \mid E_p, k_p)$, whose average over
the prompts that show a paper is its citation rate when shown; and the
exposure probability $s_t = m / N_t$, which falls as the catalogue grows
(we drop the subscript and write $s$ when the round is fixed). The
concentration and exclusion identities use the randomized design and
conditional prompt independence, not the choice law itself; the rarity
theorem of Section~\ref{sec:recursion} additionally uses the weighted
fixed-budget law and the common-lift model.

\textbf{The null.}
The benchmark's baseline is the uniform special case of the choice law
itself, which we state once and use everywhere.
\begin{proposition}[Protocol-matched null]
\label{prop:null}
If all shown weights are equal, every $k_p$-subset of the panel is equally
likely; consequently $\Pr(j \in C_p \mid j \in E_p) = k_p / m$ for every
shown paper, and the unconditional citation probability of any catalogue
paper is $k_p / N_t$.
\end{proposition}
The null is not an external baseline. It is the same model on the same
panel with the same realized budget, with preference switched off, and
every excess in this paper therefore measures a model's distance from itself
under indifference. Operationally, we compute the null by uniform redraws
that inherit each prompt's panel and budget, which matches the protocol,
including the variation in $k_p$ across prompts.

\begin{assumption}[Conditional prompt independence]
\label{assump:indep}
Given the panels $\{E_p\}$ and the realized budgets $\{k_p\}$, the
bibliographies $\{C_p\}$ of distinct prompts are independent.
\end{assumption}
This is plausible by design, since every prompt is a separate call with no
shared context. It is the entire assumption budget for the accounting
identities and the exclusion bound; the rarity result of
Section~\ref{sec:recursion} adds a weighted fixed-budget choice law, i.i.d.
saliences from a common law, and a multiplicative class lift. The
concentration identities of Section~\ref{sec:round0} remain exact without
it, up to a covariance remainder derived in App.~\ref{app:proofs}. The
choice law itself is a modeling language rather than an assumption, and its
testable consequences are checked where they are used; the exchangeability
and orthogonality conditions that appear later are hypotheses under test,
not assumptions. Six formal statements carry the paper, each stated at its
point of use; App.~\ref{app:proofs} restates the entire system in logical
order, with proofs and a crosswalk.

\textbf{Random exposure and a binding budget.}
Every catalogue paper is equally likely to be seen, because panels are drawn
uniformly rather than ranked by a retriever, so concentration measured
against the null is selection and not visibility. The budget matters just as
much. In uncapped pilot runs the models cite 22--27 of the 30 shown papers
and are statistically indistinguishable from the null, so scarcity is what
turns a latent preference into competition among candidates, and it mirrors
the practical constraint that a written review can carry only so many
references.

\textbf{Compliance.}
We regenerated round-0 runs until every model produced cap-compliant
bibliographies. The Claude models could not sustain cap compliance across
recursion rounds and are excluded from it, since truncating an over-budget
bibliography would censor the choice we are trying to measure rather than
constrain it (App.~\ref{app:compliance}). The cross-vendor claims of
Section~\ref{sec:round0} therefore rest on round 0, where all three vendors
are present.

\textbf{Terminology used throughout.}
A \emph{panel} is the 30 papers a prompt is shown, and the \emph{budget}
$k_p$ is the number it validly cites, capped at 10. The \emph{null} replays
the same prompt with preference switched off, drawing $k_p$ papers uniformly
from the same panel, and every \emph{excess} here is a model's value minus
the null's on identical prompts. Writing $c_j$ for the citations paper $j$
receives and $n_j$ for its exposures, a selector's \emph{map} is the vector
over the 120 papers with $j$th entry $c_j/n_j$ minus the same ratio under the
null; the \emph{consensus map} averages that vector over the eleven models.
Concentration is reported as the top-decile share, the fraction of all
citations going to the twelve most-cited papers, and as the
Herfindahl--Hirschman index $\HHI = \sum_j \bigl(c_j / \sum_i c_i\bigr)^2$,
which equals $1/120 = 0.0083$ when citations spread evenly and rises as they
bunch. Correlations are Pearson correlations between two maps, taken across
the 120 papers. The 120 real papers are \emph{seeds}, papers the models write
during the recursion are \emph{generated}, and one pass of 120 prompts is a
\emph{round}, with round 0 the static benchmark and rounds 1--11 growing the
catalogue. App.~\ref{app:notation} lists every symbol used in the paper.


\section{Round 0: What the Models Do}
\label{sec:round0}

The static benchmark (round 0) puts eleven models from three vendors and
eight domain experts in front of the same blinded prompts, with
byte-identical panels, the experts completing a 53-prompt subset of the
models' 120. We ask first whether the models concentrate citation attention
where experts given identical stimuli show no detectable counterpart
(\S\ref{sec:collapse}), then whether the model deviations amount to one
shared map and how much of it mixing can remove (\S\ref{sec:onemap}), and
finally what that map is made of (\S\ref{sec:mapcontent}).

\begin{table}[t]
    \centering
    \caption{Round-0 results on the full 120-prompt corpus (the teaser's
    bars are the 53-prompt expert slice). Excess concentration and exclusion
    for all eleven models against the matched null (150 uniform redraws
    inheriting each prompt's panel and budget). Reliability is the
    split-half reliability of the map (prompts split at random into halves,
    the two half-maps correlated, Spearman--Brown corrected, averaged over
    60 splits; App.~\ref{app:mapest}); $\rho$ is the noise-corrected
    alignment ratio of Proposition~\ref{prop:map} (estimator in
    App.~\ref{app:rho}).}
    \label{tab:round0}
    \small
    \begin{tabular}{@{}lccccc@{}}
        \toprule
        Model & Top-10\% share (\%) & HHI & Never cited (of 120) & Reliability & $\rho$ \\
        \midrule
        Gemini 2.5 Pro & 30.2 & 0.0146 & 15 & 0.92 & 0.98 \\
        Gemini 3.1 Pro & 28.8 & 0.0157 & 18 & 0.95 & 0.91 \\
        GPT-5 mini & 28.5 & 0.0157 & 15 & 0.95 & 0.91 \\
        Claude Sonnet 4.6 & 28.3 & 0.0156 & 20 & 0.95 & 0.95 \\
        GPT-5 & 27.8 & 0.0154 & 19 & 0.95 & 0.89 \\
        Gemini 2.5 Flash & 27.5 & 0.0141 & 7 & 0.92 & 0.91 \\
        Claude Opus 4.8 & 27.1 & 0.0143 & 11 & 0.93 & 0.98 \\
        Gemini 3.1 Flash-Lite & 25.4 & 0.0124 & 3 & 0.88 & 0.95 \\
        GPT-4.1 & 25.2 & 0.0129 & 4 & 0.89 & 0.95 \\
        GPT-4.1 mini & 24.3 & 0.0127 & 5 & 0.89 & 0.89 \\
        Claude Haiku 4.5 & 23.3 & 0.0120 & 2 & 0.86 & 0.91 \\
        \midrule
        Matched null & 15.6 & 0.0091 & 0 & -- & -- \\
        \bottomrule
    \end{tabular}
\end{table}

\subsection{Concentration, and Who Does Not Show It}
\label{sec:collapse}\label{sec:humans}

The top decile of papers takes
23.3--30.2\% of all citations against 15.6\% for the null, $\HHI$ runs
0.0120--0.0157 against 0.0091, and between 2 and 20 of the 120 papers go
uncited across roughly 30 exposures apiece where the null leaves essentially
none untouched (per-model values in Table~\ref{tab:round0}). None of this is
fabrication, since the mean per-prompt hallucination rate reaches at most
1.2\% for any model and every hallucinated reference is removed before
analysis. Concentration also tracks capability, though only loosely. The
four lightest and oldest models sit nearest the null on all three metrics,
and both Gemini pairs order as capability would predict, but the top of the
table is mixed: GPT-5 mini exceeds GPT-5 on both concentration metrics, and
Claude Sonnet 4.6 exceeds Claude Opus 4.8 on all three.

\textbf{The experts.} Whether any competent reader of the same abstracts
would produce the same deviations is an empirical question, so we put it to
eight domain experts, who completed 53 of the 120 prompts through a web
interface seeing the panels the models saw, with the same 30 blinded records,
the same fabricated metadata, and the same cap of 10. Every human-model
comparison runs on that 53-prompt slice, with model statistics recomputed on
the identical slice and the null inheriting each realized budget. Annotation
effort concentrated in one principal annotator, which we return to below
(protocol in App.~\ref{app:protocol}). On the coarse statistics the
pooled experts sit at the null, with an 18.9\% top-decile share against
18.7\% and 2 papers uncited against 1.6 expected, while on the identical slice every
model sits far above on both margins (24.5--31.9\%; 8--26 exclusions), as
Figure~\ref{fig:teaser} showed.

The two populations were not doing quite the same job, however. The experts
selected papers, while the models wrote reviews with citations embedded in
prose, so the gap could in principle belong to the task rather than to the
selector. To settle that we reran all eleven models on the same 53 prompts in
selection-only mode, asking for a citation list and no prose, with panels,
blinding and cap unchanged. Concentration does not soften. Top-decile share
runs 22.9--34.5\% and exclusions 6 to 53, against 24.5--31.9\% and 8 to 26 in
review mode, and every model stays above both the expert value and its own
null of roughly 18.2\% (App.~\ref{app:selonly}). Whatever separates the two
populations, it is not that one of them had to write an argument.

\begin{wrapfigure}{r}{0.42\textwidth}
\vspace{-1.1\baselineskip}
\centering
\includegraphics[width=\linewidth]{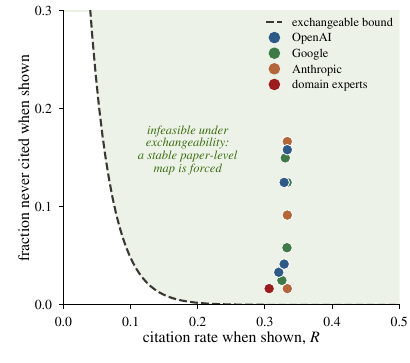}
\caption{\textbf{The same theorem, opposite verdicts.}
Proposition~\ref{prop:exch} on both populations: every model (round 0, all
120 prompts; bound drawn for that protocol) lies deep in the infeasible
region, while the pooled experts, tested at their own coverage
(App.~\ref{app:symmetry}), are consistent with exchangeability
($p = 0.42$). Colors give the vendor.}
\label{fig:symmetry}
\vspace{-2.2\baselineskip}
\end{wrapfigure}
The comparison can be sharpened into an exact test. Call a selector
\emph{exchangeable within a group} $A$ of papers if each prompt may carry
its own generosity but treats the group's papers symmetrically,
$r_{pj} = r^{A}_{p}$ for every $j \in A$; write $R_A$ for the group's
average citation rate when shown. Exchangeability allows each prompt its own generosity and forbids only stable
paper-level favoritism, and with each paper shown roughly thirty times that
absence all but guarantees every paper at least one citation.

\begin{proposition}[Exchangeability upper bound]
\label{prop:exch}
Under conditional prompt independence
(Assumption~\ref{assump:indep}) and exchangeability within $A$, the
probability that a paper in $A$ is never cited despite being shown
satisfies
\begin{equation}
    \Excl^{\mathrm{ex}}_{A}
    \;\le\;
    \frac{(1 - s R_A)^{M} - (1-s)^{M}}{1 - (1-s)^{M}},
    \label{eq:exch}
\end{equation}
with equality iff all $r^{A}_{p}$ are equal (Jensen on
Lemma~\ref{lem:excl}, the never-cited probability). At round~0, $s = 0.25$
and $M = 120$.
\end{proposition}

At the observed $R_A \approx 0.33$ that bound sits near $3\times10^{-5}$, and
every model's (rate, exclusion) pair lies deep inside the forbidden region
(Figure~\ref{fig:symmetry}), so a stable paper-level map is mathematically
forced for each of them. Because papers within a prompt compete for one
budget, we price the observed counts by redrawing each realized panel under
its own budget rather than with a binomial tail, which preserves that
dependence exactly. Across 4{,}000 draws the null exclusion count never
exceeds one while the models range from 2 to 20, so every model falls outside
the null's entire simulated support. The pooled human population, tested
identically, sits inside it (2 observed against a null mean of 1.6 and a
99.9th percentile of 6, $p = 0.48$), as it does under the analysis pipeline's
$\chi^2$ ($121.5$, df~119, $p = 0.42$). Non-rejection is not absence, so we
bound rather than assert. Matched on the same 53 prompts, panels and budgets,
expert preference strength carries a protocol-matched sensitivity bound of
7.5\% of the average model's, and the least concentrated single model exceeds that limit by a
factor of 7.6.

The pre-specified secondary outcomes point the same way. Human-model
bibliography overlap with each frontier model sits at chance, 2.96--3.25
common citations against roughly 3.06 expected with Fisher $p$ between 0.90
and 1.00, while the same models overlap one another at twice chance on the
same slice. The split-half reliability of the pooled human map is
indistinguishable from zero, $0.03 \pm 0.13$ against 0.86--0.95 for every
model (the Reliability column of Table~\ref{tab:round0}). And the raw stable
rank of the human prompt-level matrix, at 21, exceeds that of every model,
which run from 8.8 to 16.3, while its noise-corrected signal energy falls
under 1\% of raw against 5--11\% for the models. Within this benchmark's
resolution the experts are not carrying a broader map than the models.

What the models carry is therefore not relevance that any competent reader
would extract from the same abstracts, since readers given exactly those
abstracts extract at most a small fraction of it. The human map projects onto the
consensus direction with a coefficient of 0.12 and correlates with it at
0.19, so under 4\% of the variance is shared, an estimate that is noisy
by construction and directional only (App.~\ref{app:secondary}). Coverage and the dominance of one annotator bear
on that projection rather than on the exclusion test the conclusion rests on,
and dropping the dominant annotator leaves the picture unchanged
(App.~\ref{app:robustness}). Two further limits are worth keeping in view:
blinded selection is not citation practice in the wild, and eight readers are
plural in a way that a single model is not, which is the comparison a claim
about shared preference requires. The deviations are large, they are stable,
and the experts do not reproduce them. Whether they are one thing is the next
question.

\subsection{One Map, and What Mixing Can Do to It}
\label{sec:onemap}

\begin{wrapfigure}{r}{0.60\textwidth}
\vspace{-1.1\baselineskip}
\centering
\includegraphics[width=\linewidth]{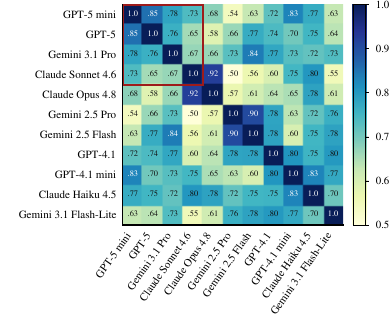}
\caption{\textbf{One map.} Correlations of the eleven maps, frontier tier
boxed. The mean off-diagonal correlation is 0.65 uncorrected and 0.71 after
correction for estimation noise, and a single principal component carries
68\% to 73\% of the cross-model variance on the same two conventions; the
second carries under 10\%.}
\label{fig:onemap}
\vspace{-0.9\baselineskip}
\end{wrapfigure}
Estimating each model's map nonparametrically and correcting it for
estimation noise with split-half reliabilities (App.~\ref{app:mapest}), the
mean correlation between the eleven maps comes to 0.65 uncorrected and 0.71
corrected, with a single principal component carrying 68\% to 73\% of
cross-model variance on the same two conventions (Figure~\ref{fig:onemap}).
We report both ends because the correction estimates within-prompt sampling
noise from pooled per-paper rates, which also carry between-prompt variation
in propensity, so it removes somewhat more than noise and the corrected value
is an upper bound. Nothing below turns on which end of that interval is used.
Vendor boundaries barely matter either, with within-vendor correlations of
0.77 for GPT, 0.79 for Gemini and 0.83 for Claude only modestly exceeding the
cross-vendor 0.70, 0.70 and 0.63. Eleven deviations are not eleven
idiosyncrasies.

The agreement shows up decision by decision. Of the shown candidates, 46\%
are cited by at most 2 of the 11 models, 2.5\% draw unanimous citation, which
has probability $5\times10^{-6}$ under independent selectors with the same
marginals, and 19\% are cited by none. Plugging the estimated per-paper
rates into the choice law reproduces this two-ended vote histogram at
total-variation distance 0.11, against 0.43 for independent selectors
(figure and construction in App.~\ref{app:votehist}). Orthogonal per-model
maps would fill the middle; shared ones pile mass at the ends, and they do.

The same conclusion holds within each model, in the theorem's own
vocabulary. For prompt $p$, let $d_p \in \mathbb{R}^{N_t}$ collect
$\delta_{pj} = s\,(r_{pj} - k_p/m)$ over the catalogue, that is, each
paper's citation probability under prompt $p$ minus the null's, the
prompt's own deviation from indifference. Its natural estimate puts
$\hat r_{pj} - k_p/m$ on the shown panel and zero elsewhere
(App.~\ref{app:rho}); the map of Section~\ref{sec:benchmark} is the
average of these vectors up to the common exposure factor $s$, and the
theorem asks how much more the $M$ of them have in common.

\begin{proposition}[Shared map and preference spectrum]
\label{prop:map}
Let $D \in \mathbb{R}^{M \times N_t}$ stack one selector's deviation
vectors $d_p$, with mean row $\bar{\boldsymbol\delta}$ and idiosyncratic part
$B = D - \mathbf{1}\bar{\boldsymbol\delta}^{\!\top}$, and let $K$ be the
total citation count. Under conditional prompt independence
(Assumption~\ref{assump:indep}),
\begin{equation}
    \mathbb{E}[\HHI] - \mathbb{E}[\HHI^{\mathrm{null}}]
    \;=\;
    \frac{M(M-1)\,\|\bar{\boldsymbol\delta}\|_2^2 - \|B\|_F^2}{K^2}
    \;=\;
    \frac{\|D\|_F^2}{K^2}\,\bigl(M\rho - 1\bigr),
    \label{eq:sharedmap}
\end{equation}
where $\rho$ is the fraction of preference strength on the shared direction and
$r_s$ the stable rank, with $\rho \le 1/r_s(D)$.
\end{proposition}

In words, a selector's excess concentration is its preference strength times
its alignment. One factor measures how much preference the prompts carry in
total, the other how much of it points in a single shared direction, and the
stable rank tells us whether it does.

After noise correction, $\rho$ runs 0.89--0.98 for every model against
$1/M = 0.008$ for unaligned prompts, and the corrected spectral ratio, the
sample counterpart of the stable rank, runs 0.85--1.06, indistinguishable
from one (App.~\ref{app:rho}), so the data attain the bound rather than
merely satisfy it, with the caveat noted there that the corrected matrix is
not guaranteed positive semidefinite. A model's 120 prompts do not carry
120 preference maps. They carry one paper-level map, applied nearly
identically across panels and across the two prompt formulations.

The identity also serves as an internal check. Reconstructing each model's
excess HHI from $(\|D\|_F^2, \rho)$ via Eq.~\eqref{eq:sharedmap}
reproduces the observed excess to within the null-redraw error, at 1.6\% mean relative deviation, for all eleven models
(App.~\ref{app:rho}). The same excess can also be read directly off pairs of bibliographies: by
Lemmas~\ref{lem:conc}--\ref{lem:overlap},
\begin{equation}
\mathbb{E}[\HHI] - \mathbb{E}[\HHI^{\mathrm{null}}]
\;=\;
\frac{2}{K^2}\sum_{p<q}\Bigl(\mathbb{E}\,|C_p \cap C_q| \;-\; \frac{k_p k_q}{N_t}\Bigr),
\label{eq:overlapmain}
\end{equation}
so excess concentration and excess pairwise bibliography overlap are the
same quantity, and the floor below is measured in the same units.

\textbf{What mixing can remove.} Which mitigation law governs depends on the
geometry. Were family maps pairwise orthogonal, excess concentration would
fall as $1/L$ in the number of families $L$, the orthogonal ideal of
App.~\ref{app:proofs}. At cross-vendor correlation near 0.70 that premise
fails decisively, so the operative law is the same identity applied across
families.

\begin{corollary}[Plurality with a shared floor]
\label{cor:floor}
If the family maps are $d_\ell = \boldsymbol\delta + \boldsymbol\varepsilon_\ell$
with residuals pairwise orthogonal, orthogonal to $\boldsymbol\delta$, and
of common energy $\varepsilon^2$, then a uniform mixture over $L$ families
satisfies
\begin{equation}
    \mathbb{E}[\HHI] - \mathbb{E}[\HHI^{\mathrm{null}}]
    \;=\;
    \underbrace{\frac{M(M-1)\,\|\boldsymbol\delta\|_2^2}{K^2}}_{\text{shared floor}}
    \;+\;
    \underbrace{\frac{M\bigl(\tfrac{M}{L}-1\bigr)\,\varepsilon^2}{K^2}}_{\text{model mixing-removable}}.
    \label{eq:floor}
\end{equation}
\end{corollary}

In words, mixing averages away only what the models do not share. The common
map passes through every mixture untouched, so mixing can pay down the
idiosyncratic part of the excess and nothing more.

The average single model carries 0.496 excess pairwise overlap and the
eleven-model, three-vendor mixture still carries 0.356, so uniform mixing
removes 28\% of the excess. What remains is the uniform-mixture residual
rather than the shared floor $\|\boldsymbol\delta\|_2^2$ of
Corollary~\ref{cor:floor}, since at finite $L$ it still contains within-family
residual alongside the shared term, which is why an optimized mixture can go
below it. The corollary's floor is derived for uniform weights rather than
for the best ones, so we solve for the optimum directly. Minimizing a
mixture's excess overlap over the simplex of convex weights leaves 52.4\% of
the average single model's excess against 71.7\% for the uniform mixture.
Because those weights are chosen on the same prompts they are scored on, we
also cross-fit them: optimizing on a random half of the prompts and
evaluating on the other half, over 80 half-splits, gives 54.6\% (95\% CI
48.9--60.6\%). The optimism is 2.2 points, so more than half the excess
survives the best mixture whether or not the weights are honest.

The optimum is not diversification. It puts all its weight on the three least
concentrating models and none on the other eight, and the best single model
alone already reaches 59.7\%, so mixing buys 7.3 points beyond simply
choosing that model. Plurality is therefore worth less than the uniform
comparison implies, the attainable floor is lower than it implies, and more
than half the excess survives the best mixture this population admits. What
the floor is not has already been settled in Section~\ref{sec:collapse},
namely relevance that experts would extract; what remains is to say what it
is.

\begin{figure}[t]
    \centering
    \includegraphics[width=\textwidth]{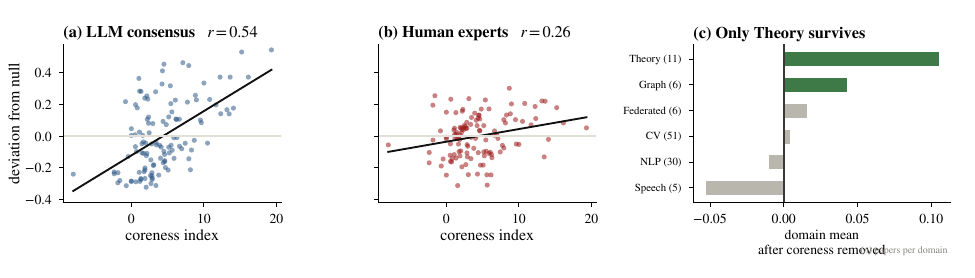}
    \caption{\textbf{A methodological-centrality prior experts barely
    share.} (a)~Consensus deviation against the coreness index, $r = 0.54$;
    (b)~the human map against the same index, $r = 0.26$, matched axes.
    (c)~Domain means after removing the coreness component; Theory and Graph
    keep positive tilts, on 11 and 6 papers.}
    \label{fig:mapcontent}
\end{figure}

\section{What the Map Is Made Of}
\label{sec:mapcontent}

\subsection{What the Map Tracks}
\label{sec:maptracks}

A cross-validated model on title and abstract text predicts a held-out
paper's consensus deviation at a correlation of 0.47 out of sample; since the
displayed
author names and years are fabricated, the map is a stable function of what the abstract
says, not noise and only marginally the displayed author and year fields,
whose contribution Section~\ref{sec:mapistext} quantifies
(App.~\ref{app:mapchar}); whether a model recognizes a paper from its text is
taken up in Section~\ref{sec:mapistext}. Two post-hoc probes indicate what
that text function tracks. A methodological-coreness keyword index correlates
0.54 with the consensus map and only 0.26 with the human map, and the
extremes are concrete, with ``Distilling Knowledge by Mimicking Features'' at
the favored end and a 3D object-detection paper at the shunned end
(Figure~\ref{fig:mapcontent}a,b, a steep line against a flat cloud on matched
axes). Coreness alone explains $R^2 = 0.29$ of the consensus map, domain
labels 0.17, and both together 0.40; once coreness is controlled for, Theory
and Graph retain positive tilts and Speech a negative one, on domain counts
of 11, 6 and 5 papers respectively (Figure~\ref{fig:mapcontent}c). Both
probes are correlates rather than the estimate the section rests on, which is
the cross-validated text model. Taken together they describe a
methodological-centrality prior that favors core-method and theoretical work
and starves applied work, applied far more sharply than the experts apply the
same tilt.

\subsection{Whether It Is the Text}
\label{sec:mapistext}

\begin{wrapfigure}{r}{0.42\textwidth}
\vspace{-1.0\baselineskip}
\centering
\includegraphics[width=\linewidth]{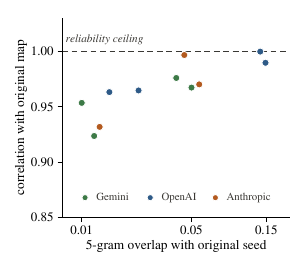}
\caption{\textbf{The map survives paraphrase.} GPT-5 mini rerun on eleven
paraphrased versions of the seed corpus; each point is one paraphraser,
placed by its residual 5-gram overlap with the original text (log scale)
against the noise-corrected correlation of its map with the original-seed
map. The correlation stays at the ceiling across a fourteen-fold range of
surface retention; the two points near 0.92 are the arms whose paraphrases
shifted meaning (App.~\ref{app:paraphrase}).}
\label{fig:paraphrase}
\vspace{-0.8\baselineskip}
\end{wrapfigure}

Predictability from text leaves the sharpest objection standing. The 120
seeds are real, public papers; every model may have seen them, and their
real citation records, during training, so a text-predictable map could
still be a memorized citation consensus using the abstract as a lookup
key. The benchmark's blinding cannot answer this, because the one field it
must keep real, the abstract, is exactly the field a memory would match
on. So we attack the surface directly. Each of the eleven models rewrote the
title and abstract of all 120 seeds, and GPT-5 mini then replayed round 0
against every rewritten corpus, with panels, prompts, and randomization
byte-identical to the original run. A valid paraphrase has to do two opposing
things at once, destroying the wording a memory could key on, with mean
5-gram overlap against the original at most 0.05, while preserving the
content the map is supposed to read, with coreness correlation above 0.95 and
length within 15\%. Four paraphrasers pass both tests
cleanly; App.~\ref{app:paraphrase} reports all eleven.

The map does not move. Against the four valid rewrites, the noise-corrected
correlation between the paraphrased-seed map and the original-seed map is
0.96--1.00, at the reliability ceiling, and 0.92--1.00 across all eleven
arms, spanning a fourteen-fold range of residual surface overlap
(Figure~\ref{fig:paraphrase}); the four clean rewrites, produced by four
models from three vendors, also agree with one another at 0.95--0.99, so the
invariant is the seeds' content rather than any rewriter's style. The
collapse itself is equally unmoved, with top-decile shares of
27.9--28.9\% against 28.5\% on the original corpus. And the two arms that
do move the map, the pair near 0.92, are exactly the two whose rewrites
shifted meaning rather than wording, which functions as a positive control.
Perturb the content and the map responds; destroy the surface and it does
not. Whatever the map reads, it reads through meaning. Paraphrase
alone cannot exclude a model recognizing a paper semantically, but that
residual channel must itself operate through content, and
Section~\ref{sec:recursion} shows the collapse on generated papers that no
training corpus contains.

Paraphrase holds the sampling design fixed, which leaves a second family of
explanations standing. Five things reach the model for each record: the
surface wording, the meaning, the fabricated surname and assigned year, the
panel it appears in and therefore the competitors it faces, and its position
in that list. One random seed bound all five to the record together, so any
of the four that paraphrase preserves could be producing the stable
per-record rate. Three further experiments separate them, all at round 0 and
all measured against the same matched null (App.~\ref{app:design}).

\textbf{Separating the five.} The first experiment moves the text and leaves the design alone.
Six shuffles in which no record keeps its own text, formally derangements,
relocate every title and abstract to a different record, which keeps its own
surname, year, panel memberships and positions; before spending we verified
that membership, order and every assigned surname and year were
byte-identical to the reference run, with title and abstract text the only
difference in the rendered prompts. Writing $d^0$ for the reference map and
$d^A_a$ for arm $a$ under permutation $\pi_a$, we regress $d^A_a[i]$ on
$d^0[\pi_a(i)]$, the propensity of the content that arrived, against
$d^0[i]$, the propensity of the slot that stayed. The two regressors are the
same vector under different index orderings, so they carry equal variance
and are near-orthogonal for a random derangement, and their coefficients
read as variance shares. At 480 prompts per arm, about 120 exposures per
record, the content share is 0.97 (95\% CI 0.92--1.03). Citation propensity
travels with the abstract rather than with the record it was attached to,
and this conclusion requires no claim that we have correctly enumerated what
the model sees.

The second experiment inverts it. The generator assigning panel membership,
within-panel order, metadata and the task split was separated into four
independent streams, and six replicates move all four at once while the
corpus is untouched, leaving the text as the only property that persists
across replicates. The six maps agree at a mean pairwise correlation of
$\bar r = 0.894$ (record bootstrap 0.863--0.919) against a generation-noise
ceiling of 0.973 measured from repeated generations on identical panels.
Because $\bar r$ is the correlation between two maps sharing only their
text, it is itself the content share of observed map variance.

The third measures each remaining factor on its own. Against a ceiling of
0.993 at 480 prompts, redrawing metadata alone decorrelates the map to
0.938, while redrawing within-panel order alone leaves it at 0.989.
Converting each gap to a variance share and combining with the panel
contribution recovered from split-half gives Table~\ref{tab:channels}. The
non-content factors there sum to the shortfall of $\bar r$ below the
generation-noise ceiling, so the accounting closes.

\begin{table}[t]
\centering
\caption{\textbf{Composition of the round-0 map}, GPT-5 mini, as shares of
its observed variance at two designs. Metadata and order are measured
directly by redrawing one stream at a time; panel is recovered from
split-half, generation from repeated generations on identical panels.
Each measured gap is converted to a noise-to-signal ratio before being
placed in a column, so that a factor measured at one exposure count can be
stated at another. Sampling factors shrink with exposures per record; the
metadata factor does not. Columns sum to 100 up to rounding.}
\label{tab:channels}
\begin{tabular}{lrrl}
\toprule
factor & 120 prompts & 480 prompts & kind \\
\midrule
content             & 89.4\% & 93.1\% & \\
fabricated metadata & 5.3\%  & 5.5\%  & structural \\
generation noise    & 2.5\%  & 0.6\%  & sampling \\
within-panel order  & 1.5\%  & 0.4\%  & sampling \\
task assignment     & 1.1\%  & 0.3\%  & sampling \\
panel composition   & 0.3\%  & 0.1\%  & sampling \\
\bottomrule
\end{tabular}
\end{table}

The factors are of two kinds. Panel, order, task and generation are sampling
artifacts, and quadrupling the exposures per record shrinks each by roughly a
factor of four, as it should; within-panel order is the largest of them and
panel composition the smallest, at 1.5\% and 0.3\% of map variance. Metadata
does not move, because a surname is fixed to a record within a run and no
amount of averaging removes it, which makes it the only non-content factor
that survives arbitrary power and, at the benchmark's design, about equal to
all the sampling factors taken together. What drives it is not the assigned
year, which shifts a record by 0.36 map standard deviations across the full
2017--2022 range and explains under 2\% of map variance on its own. The
majority of the channel is carried by the fabricated surname, a string drawn
at random that says nothing whatever about the paper it labels. Surname identity is significant on the map after
record and replicate means are removed ($F = 2.34$, $p < 0.001$, over the 63
surnames recurring across replicates), while no coarse feature accounts for
it: string length, initial letter and patronymic prefix are all null.

\textbf{A second model.} The crossover bounds the same factor independently
rather than estimating
it. Its slot coefficient, which absorbs everything fixed to the record other
than the text, averages 0.033 with a 95\% interval of $-0.001$ to 0.068
across arms, and no individual arm's coefficient differs from zero. Nothing
detectable travels with the slot, and the metadata arm's 0.055 sits inside
that bound. We repeated the crossover and the metadata arm on GPT-4.1 mini against
byte-identical designs, reading each model against its own ceiling since the
two differ by a factor of 3.4 in variance-share units. The content share
replicates, at 0.98 against 0.96, each inside the other's interval. The
metadata channel does not: 5.8\% against 2.6\%, a difference twice the
resolution threshold we fixed in advance, with the stronger model carrying
the larger name effect. Both routes recover that ordering, at ratios of 2.2
and 2.0. Two models do not establish a trend, and we report the ordering
rather than a scaling claim.

The concentration itself is unmoved throughout. Across the six resampled
designs the top-decile share runs 27.3--30.2\% against 28.5\% on the
original and top-decile excess over the matched null runs 0.117--0.145
against 0.129, with each replicate's null recomputed against its own panels.
Concentration is a property of the population of designs rather than of one
realized draw.

Where the content function comes from remains genuinely open: an absorbed
citation consensus that survives paraphrase because it was learned at the
level of meaning, or a convergent judgment of how a citable paper reads.
Every result in this paper holds under either, and
Section~\ref{sec:discussion} takes up what still separates them. Round 0
thus establishes a single, stable, content-driven centrality prior that
experts do not share and that does not depend on the words it arrives in;
the next section establishes what repeated use does with it under
scarcity.

\section{Recursion: A Fixed Preference Under Dilution}
\label{sec:recursion}\label{sec:persists}\label{sec:funnel}\label{sec:competition}\label{sec:capability}

Round 0 measured the map in a static world, a fixed corpus of real papers
judged once. That is not how the technology will be used. We therefore run
the benchmark recursively with eight models (four OpenAI, four Gemini;
Section~\ref{sec:benchmark} explains the Claude exclusion). Each round's
120 generated papers join the next round's catalogue, panels stay uniformly
random, and by round 11 the original papers are outnumbered roughly ten to
one. Being cited never raises a paper's chance of being shown again, of
surviving, or of having successors, so nothing here is a citation feedback
loop in the selection sense; what recursion supplies is dilution.
Hallucination stays marginal throughout, at a mean per-prompt rate of at most
1.4\% for any model pooled over rounds (App.~\ref{app:compliance}). The
question is what a fixed, shared map does to a literature that increasingly
consists of its own outputs.

As the catalogue grows under a fixed citation budget, even an indifferent selector
concentrates. Fewer of the papers it has ever seen can appear in any one
bibliography, so the null's top-decile share climbs from its round-0 level to 29.1\% by
round 11. Part of any concentration trajectory is therefore mechanical, and
the models must be judged against the moving bar, not against round 0.

The models clear the moving bar at every round. Round-11 top-decile shares
run 31.1--40.3\%
(Figure~\ref{fig:recursion}a), 39--49\% of shown papers are never cited by
round 11 against 38.0\% for the null, and excess bibliography overlap, the
quantity Lemma~\ref{lem:overlap} ties exactly to concentration, stays
positive at every round for every model. Distributional texture (Gini,
openness to new work, elite churn, top-1\% share) is reported in
App.~\ref{app:diagnostics}; none of the paper's conclusions rest on it.

A skeptic's first alternative is that the persistence is somehow about the
real papers themselves, their fame, their style, their age. If so, deleting
them from the ledger should delete the effect. Figure~\ref{fig:recursion}b
recomputes the excess over generated papers only, so that no property of any
real paper can contribute, and the excess survives in all eight models,
1.7--11.2 pp above the null at round 11. The excess therefore remains when
the statistic is restricted to generated records, though those records may
carry content descended from the seed panels they were written from.

Whether the \emph{shared} map survives is a separate question, and the
answer is largely no. Pooling every round, the cross-model correlation of
per-paper citation rates is 0.68 among seeds and 0.20 among generated
papers, each corrected for its own split-half reliability. Generated papers
still carry real per-model dispersion, about 61\% of the seeds' once the
sampling component is removed, so the gap is not an absence of signal and
every model does hold a stable preference over them. The models simply do not share it.
Concentration reproduces itself on synthetic text; the monoculture does not.
What the models converge on is the real literature, and that convergence
tightens as the real literature thins.

\begin{figure}[t]
    \centering
    \includegraphics[width=\textwidth]{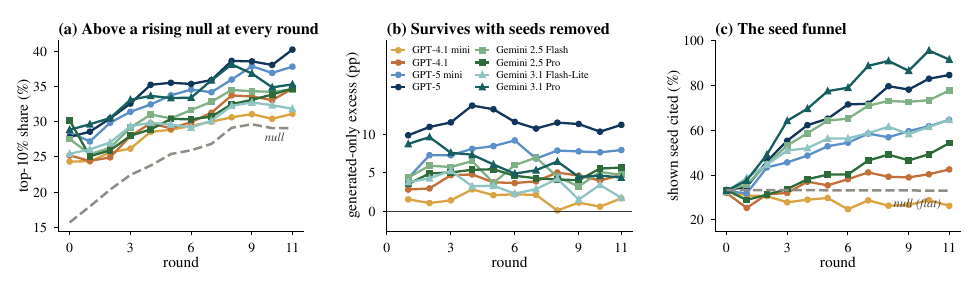}
    \caption{\textbf{Concentration under dilution} (8 models, 12 rounds).
    (a)~Top-decile citation share against the null, which itself rises as
    the catalogue grows; every model stays above it at every round. (b)~The
    excess computed over generated papers only: concentration survives with
    every real paper removed, though the cross-model agreement behind it
    does not. (c)~The seed funnel: the fraction of shown seeds cited climbs
    from 33\% toward 63\% on average (92\% for the strongest model) while
    the null's seed rate stays flat near 32\%. Panel (c) is a conditional
    rate; absolute seed attention falls over the same rounds.}
    \label{fig:recursion}
\end{figure}

\textbf{The funnel.}
As seeds become rare, the models cite them harder. The citation rate when
shown of a seed paper roughly
doubles, from 33\% at round 0 to 63\% on average at round 11, and reaches
92\% for the strongest model, which by the end cites more than nine of every
ten real papers it is shown; generated papers, meanwhile, are cited at
27--34\%, essentially the uniform rate, at every round
(Figure~\ref{fig:recursion}c). The obvious objection is that the rise is arithmetic rather than behavior,
since seeds are becoming rare and perhaps \emph{any} selector would
concentrate on them. The null answers this objection. Inheriting the same
panels and budgets, the indifferent selector's seed rate is flat near
32\% across all twelve rounds, because uniform choice
does not care which papers are rare. Everything above that flat line is
preference, and the gap, the seed excess, runs from $-7$ to $+59$ pp at
round 11; the negative end of the range is the near-liftless models
scattering around the flat line.

\textbf{Intensity is not attention.} Over the same twelve rounds the realised
number of seeds per panel falls from 30 to 2.4, against a design expectation
of 2.5, so the conditional rate and the
absolute one move in opposite directions. Seed citations per prompt fall
from 9.87 to 1.53, the unconditional probability that a given seed is cited
in a given prompt falls from 0.082 to 0.013, and the seeds' share of all
citations issued falls from everything to 15.6\%. The models cite each
surviving real paper harder and the real literature as a whole less. Both
follow from one fixed preference meeting a growing catalogue, and reporting
either alone misstates the process.

Citations do shape the catalogue, through content rather
than through exposure, since each round's papers are written from the
panels of the round before, cited and uncited alike. The trace is a cohort gradient
measurable inside a single round, where competition and panel composition
are identical across cohorts. At round 11 the seeds are cited on 63.4\% of
their exposures, the first generated cohort on 39.6\%, the second on
32.8\%, the third on 29.9\%, and every later cohort on about 28\%. The
same shape holds at round 6. No parent relation is defined here, since each
output is conditioned on a whole panel rather than on a cited subset, so this
is a cohort gradient rather than inheritance in any strict sense. It is
nonetheless the one respect in which the process feeds back on itself.

Three things change together
as the catalogue grows, and the trajectory alone cannot separate them. Real
papers become rare, which is the mechanism the next paragraphs test. Later
cohorts are cited less per exposure, as the gradient above shows, so the
generated class the seeds are compared against changes composition over
rounds. And displayed years differ between the classes by construction, since
seeds draw from 2017--2022 while generated papers draw from 2017--2025, and
Section~\ref{sec:mapistext} measured a real if small effect of displayed year.
The seed-to-generated ratio at round 1 is 1.01, essentially no gap at the
point where seeds are half the pool, which is what the last two explanations
predict and the first does not. We therefore do not attribute the magnitude of
the trajectory to rarity. What we attribute to rarity is the within-panel
result below, which compares each seed against itself inside a single round
and so holds year, cohort and quality fixed by construction.

\textbf{The mechanism.}
Why would a fixed preference produce a rising rate? Recall the choice law of Section~\ref{sec:benchmark}. A selection is a
weighted draw: each paper carries a latent weight, and a bibliography's
probability is proportional to the product of its members' weights,
Eq.~\eqref{eq:choicelaw}. We never estimate the weights individually
anywhere in this paper; they are the hypothesis space, and the map of Section~\ref{sec:round0} is the
footprint they leave. Suppose now that every seed's weight carries a common
factor $b \ge 1$ relative to the generated papers; call $b$ the
\emph{lift}. Heterogeneity within groups is allowed but carries no group information;
the idiosyncratic weights are independent draws from one common law, so the
groups differ only through $b$. The intuition is competitive. A shown
seed's chance of citation depends on who else is in its panel, and its
strongest rivals are other seeds, which carry the same lift. As seeds thin
out, each surviving seed competes mostly against generated papers it beats,
so the same fixed preference, spread over fewer eligible targets,
concentrates on each of them. Nothing about the model needs to drift.

\begin{proposition}[Heterogeneous rarity amplification]
\label{prop:rarity}
Consider a panel of size $m$ with budget $k$, drawn as in
Section~\ref{sec:benchmark}, in which a marked shown seed faces $x$ other
seeds with weights $\{b\,u_1, \dots, b\,u_x\}$ and $m-1-x$ generated papers
with weights $\{v_1, \dots, v_{m-1-x}\}$, where $b \ge 1$ and the idiosyncrasies
$u_i, v_j > 0$ are independent draws from a common nondegenerate law. Let
$r_{\mathrm{seed}}(x)$ be the resulting expected value of $r_{pj}$ for a
shown seed under the choice law \eqref{eq:choicelaw}, averaged over the
marked seed's own idiosyncrasy as well as its competitors'. Then for
$b > 1$ and $1 \le k < m$, $r_{\mathrm{seed}}(x)$ is strictly decreasing in
$x$; for $b = 1$ it is constant at $k/m$. Scarcity is required for
strictness, since at $k = m$ every shown paper is cited and the law is
indistinguishable from the null (Remark~\ref{rem:scarcity}).
Consequently, along the recursion, where a panel's seed count is
hypergeometric and stochastically decreasing in $t$ (mean $120\,m/N_t$), the
seed citation rate when shown,
$R_{\mathrm{seed}}(t) = \mathbb{E}\,[\,r_{\mathrm{seed}}(x)\,]$, is
strictly increasing in $t$ when $b > 1$ and flat at $k/m$ when $b = 1$.
\end{proposition}

\emph{In words:} the model's preference for seed papers never changes;
what changes is the panel. When a panel happens to contain fewer seeds,
each one competes with fewer other highly weighted papers for the same
citation budget, so its individual chance of being cited goes up.

The proposition's theorem is the slope; the level, $b > 1$ itself, is
read directly off the data as the gap between seed and generated citation
rates. The slope is the sharper half, since fewer seed rivals in a panel mean a
higher per-seed rate, a prediction about single panels that makes no
reference to time. The $b = 1$ case is realized in the data twice over. The null's seed
rate is flat by construction, and GPT-4.1 mini, the model with essentially
no lift, hugs that flat line throughout.

Trajectories alone cannot establish the mechanism. Between round 1 and
round 11 the catalogue, the null, and the text mix all change at once, so a
time trend could mimic the rise; the mechanism claim needs a test inside a
single moment, and the benchmark has already randomized one for us. Within
any one round, the number of seeds a panel contains is a hypergeometric
lottery; one prompt happens to draw 2 seeds, another 7, with the same
catalogue and the same round. That lottery lets us ask the mechanism's question directly. Is the same seed
paper cited more often in the panels where it happened to face fewer seed
rivals? For shown seed $i$
in prompt $p$, with $x_{-i,p}$ its number of co-shown seeds, we fit the
linear probability model
\begin{equation}
    \mathrm{cited}_{ip}
    \;=\;
    \beta\, x_{-i,p} + \gamma_{t(p)} + \alpha_i + \theta\,\mathrm{review}_p
    + \varepsilon_{ip},
    \label{eq:competition}
\end{equation}
where the fixed effects have plain readings: $\gamma_t$ compares panels
only within the same round, absorbing everything that drifts across rounds,
and $\alpha_i$ compares each seed only against itself, absorbing paper
quality, so popular seeds landing in seed-light panels by luck cannot fake
the result. We cluster errors by prompt; rounds 1--11 give 7{,}557 rows
over 1{,}292 prompt clusters per model. Round 0 is excluded because its
panels contain 30 seeds and the regressor has no variation. The proposition predicts two things about
this coefficient, that $\beta$ is negative whenever $b > 1$ and exactly zero
when $b = 1$, and both hold. Every one of the eight models has $\beta < 0$,
from $-3.07$ pp per seed competitor ($t = -11.3$) down to $-0.31$, with the
insignificant coefficients belonging to the near-liftless models
(Figure~\ref{fig:mechanism}a,b). The exact zero is verified rather than
assumed, since the identical regression on uniform redraws of the same panels
and budgets returns placebo coefficients within $\pm 0.08$ pp of zero. The full table, including the placebo and the
symmetric regression on generated papers, is App.~\ref{app:regression}.

The level and the slope are two functionals of the same parameter, estimated
from disjoint variation, one from the time-trajectory endpoint and one from
cross-panel luck inside single rounds, so a single lift driving both should
make them track each other across models. They do track each other
(Figure~\ref{fig:mechanism}c), with the Gemini 2.5 pair reversing on both
axes at once. We report the association rather than an exact rank match,
because neither functional is a monotone transform of $b$ once models are
allowed their own salience distributions, so the proposition does not predict
identical orderings and the data do not show them. The magnitude survives a
sanity check as well. Mean seed competitors fall from about 14 at round 1
to about 2.4 at round 11, so a linear $\beta = -3.07$ predicts a $+35$ pp
funnel against $+59$ observed; the shortfall has the expected sign, since
the binned response in Figure~\ref{fig:mechanism}a is convex, steepest
where competitors are few.

\begin{figure}[t]
    \centering
    \includegraphics[width=\textwidth]{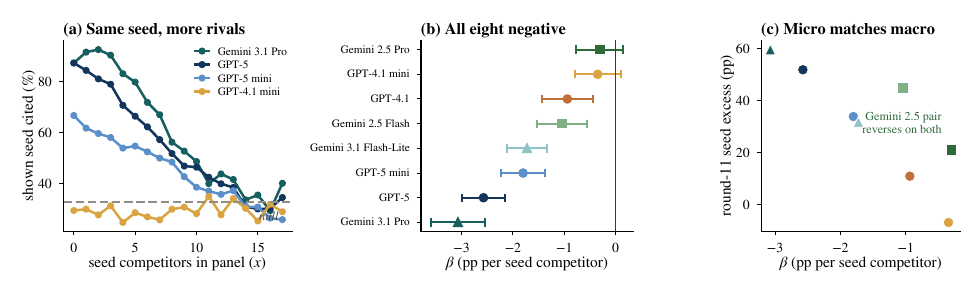}
    \caption{\textbf{The within-panel competition test.} (a)~Binned citation
    rate of a shown seed against its number of co-shown seeds, four
    representative models: monotone, convex, and flat at the null for the
    near-liftless model. (b)~$\beta$ from Eq.~\eqref{eq:competition} for all
    eight models with prompt-clustered 95\% intervals: all negative, the two
    near-liftless models not significant; the test lives inside single panels and uses no
time trend. (c)~The micro coefficient tracks the macro outcome: $\beta$
    against round-11 seed excess, the Gemini 2.5 pair reversing on both. The
    signed association is negative, equivalently positive between
    competition-effect magnitude $-\beta$ and seed excess, and is not an exact
    rank match.}
    \label{fig:mechanism}
\end{figure}

Capability tracks these quantities without governing them. Round-0
concentration and round-11 generated-only excess each order within three of
the recursion's four vendor pairs, the exception being GPT-5, whose mini
variant concentrates slightly more than the full model. The lift orders
within three of four as well, and there the exception is Gemini 2.5, where
the Pro is the more concentrated model at round 0 yet shows the weaker seed
preference.

Three limits bound what this section establishes. The seed effect is
composite, so the within-panel test isolates the competition mechanism rather
than the source of the lift, a confound taken up in
Section~\ref{sec:discussion}. The recursion runs on two vendors, so the
cross-vendor monoculture claims rest entirely on round 0, where all three are
present. And the trajectory itself, as set out above, does not identify
rarity as the driver of its own magnitude.

The picture at the end of the recursion is narrower than the trajectory
alone suggests. A shared preference over the real literature, meeting a
catalogue that dilutes, drives rising conditional selection onto a
shrinking real subset while total attention to that subset falls. The
shared component does not reproduce itself on the synthetic papers the
process generates, though each model's own concentration does, and
citability persists across a few cohorts before it washes out.


\section{Related Work}
\label{sec:related}

\textbf{Concentration in science.}
Cumulative advantage and the Matthew effect are the canonical account of
citation inequality: recognition flows disproportionately to already
recognized work and authors \citep{merton1968matthew, price1976general}, and
experiments in artificial cultural markets show that social signals alone can
manufacture inequality among items of equal quality
\citep{salganik2006experimental}. Under scarce attention
\citep{simon1971designing}, growing literatures concentrate reading onto
canon and consensus ossifies
\citep{evans2008electronic, chu2021slowed, nielsen2021global}.
Discrete-choice models of network formation bring the phenomenon down to the
level of individual decisions, modeling edge creation as a choice among local
alternatives rather than a global attachment rule \citep{overgoor2019choosing}.
The mechanisms studied in this literature are social, status, visibility, and
accumulated count; the benchmark studied here removes all three by
construction and asks what concentration remains.

\textbf{Self-consuming generative loops.}
Retraining generative models on their own output degrades quality and
diversity across generations
\citep{shumailov2024ai, alemohammad2024selfconsuming}, with theoretical work
locating the damage in the iterative estimation process
\citep{taori2023data, bertrand2024stability, dohmatob2024}
\ and bias amplification identified as a companion pathology of
model-induced distribution shift \citep{wyllie2024fairness}. The sharpest
open question is regime dependence: accumulating real and synthetic data
largely averts the collapse that replacement causes
\citep{gerstgrasser2024collapse, kazdan2025collapse}. The recursion studied
here is instead a selection loop with the models held fixed: what recurses is
the catalogue of papers, data accumulate rather than replace, and the
collapse concerns citation attention rather than model weights.

\textbf{AI in scientific writing and citation auditing.}
Language models are measurably present in scientific papers and peer reviews
\citep{liang2024mapping, liang2024monitoring}, and early system-level
evidence suggests they raise individual productivity while contracting the
collective focus of science \citep{hao2024expand}. The most audited failure
mode of model-generated bibliographies is fabrication
\citep{walters2023fabrication, agrawal2024language}. Citation-specific audits
go further: generated reference lists remain semantically aligned yet favor
already prominent, recent, and prestigious work
\citep{algaba2025how}
, and controlled-panel studies detect selection biases when candidate sets
are held fixed \citep{he2025who}.
\ Relative to these audits, the present benchmark blinds all metadata,
replaces naturally occurring retrieval with random exposure, measures the
full paper-level preference structure rather than a single varied attribute,
and adds domain experts judging identical panels under the identical cap.

\textbf{Popularity bias and feedback loops in recommendation and search.}
Recommender systems that repeatedly overexpose popular items narrow
consumption diversity and homogenize users over time
\citep{fleder2009blockbuster, chaney2018algorithmic,
klimashevskaia2024survey, wu2024result}, a concern about engineered
visibility that predates generative systems in the politics of search
\citep{introna2000shaping}. Early audits of generative search engines find
the same signature in synthesized answers, whose citations concentrate onto a
narrow head of sources \citep{yang2025, kirsten2025}.
These loops operate through exposure, what the system chooses to show; in the
benchmark studied here exposure is uniformly random, so the concentration
that remains is a preference-side loop that no exposure-diversification
intervention can remove.

\textbf{Algorithmic monoculture and homogenization.}
When many decision makers reuse the same algorithm, their decisions become
correlated, which can lower welfare even if each decision is individually
accurate \citep{kleinberg2021algorithmic, creel2022algorithmic} and
homogenizes outcomes for the individuals being judged
\citep{bommasani2022picking}. For language models the evidence to date is
output-side: individual models collapse onto narrow output distributions
\citep{wu2025generative}, different models produce strikingly similar
creative and epistemic content
\citep{anderson2024homogenization, padmakumar2024does, doshi2024generative,
wenger2025different}, and the projected endpoint for information ecosystems
has been named knowledge collapse \citep{peterson2025ai, wright2025epistemic}.
This paper measures the preference function behind such outputs, across
eleven models and three vendors, quantifies exactly how much of the resulting
concentration mixing models can remove, and supplies the heterogeneous
human baseline that the monoculture literature postulates but has not had on
identical stimuli.


\section{Discussion and Conclusion}
\label{sec:discussion}\label{sec:conclusion}

Three findings carry the argument. Every model tested concentrates citations
and excludes papers far beyond its own matched null, while pooled experts on
identical panels sit at that null. The deviations are one shared map rather
than eleven, so mixing can remove only what the models do not share. And
under recursion the preference itself holds fixed while the catalogue dilutes
around it, raising the rate at which each surviving real paper is cited even
as real papers' share of all citations falls, with the within-panel test
confirming the competition mechanism inside single panels and free of any
time trend.

Where the map comes from remains open; whether it exists does not. Moving a
record's text onto a different record carries the citation propensity with
it, and rewriting every title and abstract leaves the map unmoved, so if the
map is an absorbed citation consensus it was absorbed at the level of
meaning, where it becomes empirically adjacent to a convergent judgment about
what a citable paper reads like. Every result here holds under either
reading, and extending the crossover across selector models is future work.

Corollary~\ref{cor:floor} explains why the instinctive remedy is limited, and
the optimized mixture prices it. Mixing the eleven models uniformly removes
28\% of the excess, and optimizing the weights rather than assuming them
reaches 55\% of the average single model's under cross-fitting, 52\% in
sample, though
mostly by discarding the models that concentrate hardest rather than by
decorrelating those it keeps, since the best single model alone reaches 60\%.
That floor describes the current model population rather than a constant of
nature, but any decorrelation strategy has to be demonstrated against the
map, because cross-vendor agreement already nearly matches within-vendor
agreement. Exposure was uniform here by construction, so equalizing it is not
available as a remedy, and whether a policy that deliberately oversamples
low-propensity papers would help is untested. What is ruled out is the hope
that unequal exposure was the cause.

The design measures choice among papers already placed in front of a model,
on one topic and a narrow range of task formats, with the parsed bibliography treated as an
unordered set. It says nothing about retrieval, about which of a model's
suggestions an author keeps, or about what reaches the published record. The
recursion adds papers to a catalogue but never lets a citation change what is
shown next, so it studies a fixed preference meeting a growing corpus rather
than an ecosystem responding to its own output. The seed lift is composite,
mixing a preference for real work with any quality gap between real and
model-written papers, which is why the within-panel test isolates the
competition mechanism rather than the source of the lift. The concentration
identities rest on conditional prompt independence
(Assumption~\ref{assump:indep}) and remain exact without it up to a
covariance remainder derived in App.~\ref{app:proofs}; the rarity result adds
the weighted choice law and a common-lift model on top of that.

Within that scope the result is specific. Every displayed signal of status
was removed from the corpus and citation credit still condensed, onto the
same papers, for eleven models from three vendors, on abstracts whose authors
and years were invented. Domain experts judging identical panels showed no
comparable preference we could detect; a protocol-matched sensitivity
calculation puts any preference they carry at 7.5\% of the average model's. The concentration is therefore not a retrieval problem and not a
prestige problem, and hallucination audits will not find it, because every
reference involved is real. It sits in the conditional selection step rather
than in unequal visibility, which is where an intervention would have to act.

\section*{Acknowledgments}

This work was supported by ONR grant N00014-23-1-2714, DOE grant
DE-SC0020345, DOI grant 140D0423C0076, NSF Awards 2145346 (CAREER),
02133861 (DMS) and 2113904 (CCSS), and the NSF AI Institute for Foundations
of Machine Learning (IFML). This work was also supported by computing
resources on the Vista GPU Cluster through the Center for Generative AI
(CGAI) and the Texas Advanced Computing Center (TACC) at The University of
Texas at Austin.

\clearpage
\bibliographystyle{plain}
\bibliography{References}

\clearpage
\appendix

\section{Notation}
\label{app:notation}

Every symbol used in the main text, grouped by where it first appears.
Symbols that occur only inside one proposition are listed with that
proposition. Section~\ref{sec:benchmark} defines the design in prose.

\paragraph{Design and protocol.}
\begin{center}\small
\begin{tabular}{@{}llp{0.44\textwidth}@{}}
\toprule
symbol & meaning & definition \\
\midrule
$M$ & prompts per round & $M = 120$, half review and half position piece \\
$m$ & panel size & $m = 30$ papers shown per prompt \\
$p$ & prompt index & $p = 1,\dots,M$ \\
$j$ & paper index & over the catalogue \\
$t$ & round index & $t = 0,\dots,11$; round 0 is the static benchmark \\
$U_t$ & catalogue at round $t$ & 120 seeds plus all generated papers so far \\
$N_t$ & catalogue size & $N_t = 120\,(t+1)$, so $N_{11} = 1{,}440$ \\
$E_p$ & panel & $E_p \subset U_t$, drawn uniformly, $|E_p| = m$ \\
$C_p$ & bibliography & $C_p \subseteq E_p$, the parsed valid citations \\
$k_p$ & realized budget & $k_p = |C_p| \le 10$ \\
$s_t$ & exposure probability & $s_t = m/N_t$; $s = 0.25$ at round 0 \\
\bottomrule
\end{tabular}
\end{center}

\paragraph{Selection, maps, and concentration.}
\begin{center}\small
\begin{tabular}{@{}llp{0.44\textwidth}@{}}
\toprule
symbol & meaning & definition \\
\midrule
$w_{pj}$ & latent weight & the choice law's preference parameter under prompt $p$, Eq.~\eqref{eq:choicelaw}; prompt-invariance is tested in \S\ref{sec:onemap}, not assumed \\
$r_{pj}$ & citation probability & $\Pr(j \in C_p \mid E_p, k_p)$ \\
$c_j,\; n_j$ & citations, exposures & counts for paper $j$ over a run \\
map & per-paper deviation & vector with $j$th entry $c_j/n_j$ minus the null's \\
$\delta_{pj}$ & budget-neutral deviation & $\delta_{pj} = s\,(r_{pj} - k_p/m)$ \\
$d_p$ & deviation vector & $d_p \in \mathbb{R}^{N_t}$, collecting $\delta_{pj}$ over $j$ \\
$D$ & deviation matrix & $D \in \mathbb{R}^{M \times N_t}$, stacking the $d_p$ \\
$\bar{\boldsymbol\delta}$ & mean row of $D$ & the selector's shared direction \\
$B$ & idiosyncratic part & $B = D - \mathbf{1}\bar{\boldsymbol\delta}^{\!\top}$ \\
$K$ & total citations & summed over all prompts \\
$\HHI$ & concentration index & $\sum_j (c_j / \sum_i c_i)^2$; $0.0083$ under even spreading \\
top-decile share & concentration & citations to the 12 most-cited papers, as a fraction \\
$r$ & correlation & Pearson, between two maps, across the 120 papers \\
\bottomrule
\end{tabular}
\end{center}

\paragraph{Shared structure and mixing (Props.~\ref{prop:exch},~\ref{prop:map}, Cor.~\ref{cor:floor}).}
\begin{center}\small
\begin{tabular}{@{}llp{0.44\textwidth}@{}}
\toprule
symbol & meaning & definition \\
\midrule
$A$ & a group of papers & the set within which exchangeability is tested \\
$R_A$ & group citation rate & average rate when shown, over $A$ \\
$\Excl^{\mathrm{ex}}_{A}$ & never-cited probability & under exchangeability within $A$, Eq.~\eqref{eq:exch} \\
$\rho$ & alignment ratio & share of preference strength on the shared direction \\
$r_s$ & stable rank & $r_s(D) = \|D\|_F^2 / \|D\|_2^2$, with $\rho \le 1/r_s(D)$; the estimated version is a corrected spectral ratio (App.~\ref{app:mapest}) \\
$L$ & preference families & number of distinct maps in a mixture \\
$\boldsymbol\delta$ & shared component & the map common to all families \\
\bottomrule
\end{tabular}
\end{center}

\paragraph{Recursion (Prop.~\ref{prop:rarity}, Eq.~\eqref{eq:competition}).}
\begin{center}\small
\begin{tabular}{@{}llp{0.44\textwidth}@{}}
\toprule
symbol & meaning & definition \\
\midrule
$b$ & seed lift & multiplicative factor on seed weights, $b \ge 1$ \\
$u_i,\; v_j$ & idiosyncrasies & i.i.d.\ positive weights from a common law $F$ \\
$x$ & seed competitors & other seeds facing a marked shown seed \\
$r_{\mathrm{seed}}(x)$ & seed rate at $x$ & expected citation rate of a shown seed \\
$R_{\mathrm{seed}}(t)$ & seed rate at round $t$ & $\mathbb{E}[\,r_{\mathrm{seed}}(x)\,]$ over the round's panels \\
$\beta$ & competition slope & coefficient on co-shown seeds, Eq.~\eqref{eq:competition} \\
$\gamma_t$ & round fixed effect & compares panels only within the same round \\
$\theta$ & task effect & review versus position piece \\
\bottomrule
\end{tabular}
\end{center}

\section{Theory}
\label{app:proofs}

This appendix restates every formal result of the paper in one place, in
logical order, and proves it. Everything below is a consequence of the
choice law, Eq.~\eqref{eq:choicelaw}, and the design of
Section~\ref{sec:benchmark}; Assumption~\ref{assump:indep} is invoked only
where marked, and the identities that use it remain exact without it, up to
the covariance remainder $\Gamma$ made explicit below. Table~\ref{tab:xwalk}
maps each statement to where it is stated and used in the main text.

\begin{table}[ht]
    \centering
    \caption{Crosswalk between the theory and its use.}
    \label{tab:xwalk}
    \small
    \begin{tabular}{@{}llll@{}}
        \toprule
        Statement & Stated & Fired at & Proof \\
        \midrule
        Prop.~\ref{prop:null} (protocol-matched null) & \S\ref{sec:benchmark} & every ``excess'' in the paper & \S\ref{app:null} \\
        Assumption~\ref{assump:indep} (prompt independence) & \S\ref{sec:benchmark} & scope in \S\ref{app:setting} & -- \\
        Lemma~\ref{lem:conc} (concentration identity) & here & \S\ref{sec:onemap}, \S\ref{sec:recursion} & \S\ref{app:identities} \\
        Lemma~\ref{lem:overlap} (overlap equivalence) & here & \S\ref{sec:onemap}, \S\ref{sec:recursion} & \S\ref{app:identities} \\
        Prop.~\ref{prop:exch} (exchangeability bound) & \S\ref{sec:collapse} & Figure~\ref{fig:symmetry} & \S\ref{app:exclusion} \\
        Prop.~\ref{prop:map} (shared map, spectrum) & \S\ref{sec:onemap} & Table~\ref{tab:round0}, the identity check & \S\ref{app:map} \\
        Cor.~\ref{cor:orth} (orthogonal ideal) & here & \S\ref{sec:onemap}, italicized & \S\ref{app:map} \\
        Cor.~\ref{cor:floor} (shared floor) & \S\ref{sec:onemap} & the uniform-mixture comparison & \S\ref{app:map} \\
        Lemma~\ref{lem:excl} (never-cited probability) & here & inside Prop.~\ref{prop:exch} & \S\ref{app:exclusion} \\
        Prop.~\ref{prop:rarity} (rarity amplification) & \S\ref{sec:recursion} & Figures~\ref{fig:recursion}c, \ref{fig:mechanism} & \S\ref{app:rarity} \\
        \bottomrule
    \end{tabular}
\end{table}

\subsection{Setting, notation, and the assumption's exact scope}
\label{app:setting}

Fix a round $t$ with catalogue $U_t$ of size $N_t$; we drop the round
subscript on all per-prompt quantities. Prompt $p \in \{1, \dots, M\}$ is
shown a uniformly random panel $E_p \subset U_t$ with $|E_p| = m$ and cites
$C_p \subseteq E_p$ with $|C_p| = k_p$, according to the choice law
\eqref{eq:choicelaw} with positive weights $w_{pj}$. Summing the law over
admissible subsets containing a marked shown paper gives its exact inclusion
probability on a realized panel,
\begin{equation}
    \pi_j(E_p)
    \;=\;
    \Pr\!\left(j \in C_p \mid E_p\right)
    \;=\;
    \frac{w_{pj}\, e_{k_p-1}\!\left(w_{p,\,E_p \setminus \{j\}}\right)}
         {e_{k_p}\!\left(w_{p,\,E_p}\right)},
    \label{eq:inclusion}
\end{equation}
and $r_{pj} = \mathbb{E}\left[\pi_j(E_p) \mid j \in E_p\right]$ is its
marginalization over panels containing $j$. Write $X_{pj}$ and $Y_{pj}$ for
the exposure and citation indicators, so that under uniform exposure
$\Pr(X_{pj} = 1) = s = m/N_t$ and
$q_{pj} = \Pr(Y_{pj} = 1) = s\, r_{pj}$. The per-paper citation count is
$C_j = \sum_p Y_{pj}$, the total is $K = \sum_p k_p = \sum_j C_j$, and
$\HHI = \sum_j (C_j / K)^2$.

Assumption~\ref{assump:indep} states that, given the panels and realized
budgets, the bibliographies of distinct prompts are independent. It is used
only for the cleanest forms of the overlap, exclusion, and shared-map
results; where it fails, Lemmas~\ref{lem:conc} and \ref{lem:overlap} below
remain exact with the covariance remainder
\begin{equation}
    \Gamma_{pq} \;=\; \sum_{j=1}^{N_t} \operatorname{Cov}\!\left(Y_{pj},\, Y_{qj}\right)
    \label{eq:gamma}
\end{equation}
carried explicitly.

\subsection{The null}
\label{app:null}

\begin{restatement}{Proposition~\ref{prop:null} (Protocol-matched null; restated)}
If all shown weights are equal, every $k_p$-subset of the panel is equally
likely; consequently $\Pr(j \in C_p \mid j \in E_p) = k_p/m$ for every shown
paper, and the unconditional citation probability of any catalogue paper is
$k_p / N_t$.
\end{restatement}

\begin{proof}
If all shown weights in $E_p$ equal a common constant $c > 0$, then for
every admissible subset $C \subseteq E_p$ with $|C| = k_p$ the numerator of
the choice law is $c^{k_p}$. The denominator is the number of admissible
subsets times the same constant, so every admissible subset is equally
likely. The inclusion probability of any shown paper is therefore the
fraction of all $k_p$-subsets that contain it, namely $k_p/m$. Because the
exposure probability of any catalogue paper is $m/N_t$, the unconditional
citation probability is $(m/N_t)(k_p/m) = k_p/N_t$.
\end{proof}

Define the deviation of prompt $p$ at paper $j$ from its own null,
\begin{equation}
    \delta_{pj} \;=\; q_{pj} - \frac{k_p}{N_t}
    \;=\; s\left(r_{pj} - \frac{k_p}{m}\right),
    \qquad
    d_p = (\delta_{p1}, \dots, \delta_{pN_t})^{\!\top} .
    \label{eq:deviation}
\end{equation}
Because each prompt cites exactly $k_p$ papers,
$\sum_{j} \delta_{pj} = 0$ for every $p$: deviations are budget-neutral, a
fact used repeatedly below.

\subsection{The two accounting identities}
\label{app:identities}

\begin{lemma}[Covariance-aware excess-concentration identity]
\label{lem:conc}
For every round,
\[
    \mathbb{E}[\HHI] - \mathbb{E}[\HHI^{\mathrm{null}}]
    \;=\;
    \frac{2}{K^2} \sum_{1 \le p < q \le M}
    \left( \langle d_p, d_q \rangle + \Gamma_{pq} \right).
\]
In particular, under Assumption~\ref{assump:indep},
$\mathbb{E}[\HHI] - \mathbb{E}[\HHI^{\mathrm{null}}]
= \tfrac{2}{K^2} \sum_{p<q} \langle d_p, d_q \rangle$.
\end{lemma}

\begin{proof}
By definition,
$\mathbb{E}[\HHI] = \tfrac{1}{K^2} \sum_j \mathbb{E}[C_j^2]$. Because
$C_j = \sum_p Y_{pj}$,
\[
    \mathbb{E}[C_j^2]
    = \sum_{p=1}^{M} \mathbb{E}[Y_{pj}]
    + 2 \sum_{1 \le p < q \le M} \mathbb{E}[Y_{pj} Y_{qj}],
    \qquad
    \mathbb{E}[Y_{pj} Y_{qj}] = q_{pj} q_{qj} + \operatorname{Cov}(Y_{pj}, Y_{qj}).
\]
Under the protocol-matched null, each $q_{pj}$ is replaced by $k_p/N_t$.
Subtracting the null expression and expanding
$q_{pj} = k_p/N_t + \delta_{pj}$ yields, for each $j$, the linear term
$\sum_p \delta_{pj}$ and, for each pair $p < q$, the cross terms
$\delta_{pj}\delta_{qj} + (k_p/N_t)\delta_{qj} + (k_q/N_t)\delta_{pj}
+ \operatorname{Cov}(Y_{pj}, Y_{qj})$. Summing over $j$, the linear and
mixed terms vanish by budget neutrality
($\sum_j \delta_{pj} = 0$ for every $p$), leaving
\[
    \mathbb{E}[\HHI] - \mathbb{E}[\HHI^{\mathrm{null}}]
    = \frac{2}{K^2} \sum_{1 \le p < q \le M}
    \Bigl( \textstyle\sum_j \delta_{pj} \delta_{qj}
    + \sum_j \operatorname{Cov}(Y_{pj}, Y_{qj}) \Bigr),
\]
which is the claimed identity; the second form follows since Assumption~\ref{assump:indep} sets every $\Gamma_{pq}$ to zero.
\end{proof}

\begin{lemma}[Bibliography-overlap equivalence]
\label{lem:overlap}
For prompts $p \neq q$, the overlap count
$O_{pq} = |C_p \cap C_q|$ satisfies
\[
    \mathbb{E}[O_{pq}] - \frac{k_p k_q}{N_t}
    \;=\;
    \langle d_p, d_q \rangle + \Gamma_{pq},
    \qquad\text{hence}\qquad
    \mathbb{E}[\HHI] - \mathbb{E}[\HHI^{\mathrm{null}}]
    =
    \frac{2}{K^2} \sum_{p<q}
    \left( \mathbb{E}[O_{pq}] - \frac{k_p k_q}{N_t} \right).
\]
\end{lemma}

\begin{proof}
By definition and taking expectations,
\[
    O_{pq} = \sum_j Y_{pj} Y_{qj},
    \qquad
    \mathbb{E}[O_{pq}] = \sum_j \left( q_{pj} q_{qj}
    + \operatorname{Cov}(Y_{pj}, Y_{qj}) \right).
\]
Under the null,
$\mathbb{E}[O^{\mathrm{null}}_{pq}]
= \sum_j (k_p/N_t)(k_q/N_t) = k_p k_q / N_t$. Subtracting and expanding
$q_{pj} = k_p/N_t + \delta_{pj}$, the mixed terms again vanish by budget
neutrality, leaving
$\langle d_p, d_q \rangle + \Gamma_{pq}$. Substituting into
Lemma~\ref{lem:conc} gives the second display.
\end{proof}

Together the lemmas say: excess concentration is exactly accumulated excess
pairwise bibliography overlap, modulo the same remainder. One biased map and
homogenized bibliographies are one number, which is how
Section~\ref{sec:onemap} uses them.

\subsection{The shared map, the spectrum, and mixtures}
\label{app:map}

Stack the deviation vectors \eqref{eq:deviation} into
$D \in \mathbb{R}^{M \times N_t}$ with mean row
$\bar{\boldsymbol\delta} = \tfrac{1}{M} D^{\!\top} \mathbf{1}_M$,
idiosyncratic part
$B = D - \mathbf{1}_M \bar{\boldsymbol\delta}^{\!\top}$, alignment ratio
$\rho = \|D^{\!\top}\mathbf{1}_M\|_2^2 / (M \|D\|_F^2)$, and stable rank
$r_s(D) = \|D\|_F^2 / \|D\|_2^2$. The realized-panel rows used by the
estimator of App.~\ref{app:rho} (entries $\hat r_{pj} - k_p/m$ on the shown
panel, zero elsewhere) are the natural unbiased estimates of these
population rows up to the common factor $s$, which cancels in $\rho$ and
$r_s$.

\begin{restatement}{Proposition~\ref{prop:map} (Shared map and preference spectrum; restated)}
Under Assumption~\ref{assump:indep},
\[
    \mathbb{E}[\HHI] - \mathbb{E}[\HHI^{\mathrm{null}}]
    =
    \frac{M(M-1)\|\bar{\boldsymbol\delta}\|_2^2 - \|B\|_F^2}{K^2}
    =
    \frac{\|D\|_F^2}{K^2}\left(M\rho - 1\right),
    \qquad
    \rho \le \frac{1}{r_s(D)} .
\]
The bound is tight whenever $D$ has rank one with $\mathbf{1}_M$ a leading
left singular vector.
\end{restatement}

\begin{proof}
Under Assumption~\ref{assump:indep} the remainder in Lemma~\ref{lem:conc}
vanishes, so
$\mathbb{E}[\HHI] - \mathbb{E}[\HHI^{\mathrm{null}}]
= \tfrac{2}{K^2} \sum_{p<q} \langle d_p, d_q \rangle$. Now
\[
    \|D^{\!\top}\mathbf{1}_M\|_2^2
    = \Bigl\| \sum_p d_p \Bigr\|_2^2
    = \sum_p \|d_p\|_2^2 + 2\sum_{p<q} \langle d_p, d_q \rangle
    = \|D\|_F^2 + 2\sum_{p<q} \langle d_p, d_q \rangle,
\]
so the excess equals
$\left( \|D^{\!\top}\mathbf{1}_M\|_2^2 - \|D\|_F^2 \right) / K^2$. Because
$D = \mathbf{1}_M \bar{\boldsymbol\delta}^{\!\top} + B$ with
$\mathbf{1}_M^{\!\top} B = 0$, we have
$D^{\!\top}\mathbf{1}_M = M \bar{\boldsymbol\delta}$ and
$\|D\|_F^2 = M\|\bar{\boldsymbol\delta}\|_2^2 + \|B\|_F^2$; substituting
gives the first equality. The second is the definition of $\rho$. For the
bound, $\|D^{\!\top}\mathbf{1}_M\|_2 \le \|D\|_2 \|\mathbf{1}_M\|_2
= \sqrt{M}\,\|D\|_2$; squaring and dividing by $M\|D\|_F^2$ gives
$\rho \le \|D\|_2^2 / \|D\|_F^2 = 1/r_s(D)$. Tightness under the rank-one
alignment condition is immediate.
\end{proof}

\begin{corollary}[Orthogonal ideal]
\label{cor:orth}
Suppose the $M$ prompts are partitioned into $L$ equal families, each family
sharing a deviation vector $\boldsymbol d_\ell$ with the $\boldsymbol d_\ell$ pairwise
orthogonal and $\|\boldsymbol d_\ell\|_2 = a$ for all $\ell$. Under
Assumption~\ref{assump:indep},
\[
    \mathbb{E}[\HHI] - \mathbb{E}[\HHI^{\mathrm{null}}]
    =
    \frac{M\left(\tfrac{M}{L} - 1\right) a^2}{K^2},
\]
strictly decreasing in the number $L$ of families.
\end{corollary}

\begin{proof}
By orthogonality, the only nonzero inner products
$\langle d_p, d_q \rangle$ come from pairs within one family, each
contributing $a^2$; there are $L \binom{M/L}{2}$ such pairs. By
Lemma~\ref{lem:conc} under independence, the excess is
$\tfrac{2}{K^2} L \binom{M/L}{2} a^2 = M(\tfrac{M}{L}-1) a^2 / K^2$.
\end{proof}

\begin{restatement}{Corollary~\ref{cor:floor} (Plurality with a shared floor; restated)}
If the family maps are
$\boldsymbol d_\ell = \boldsymbol\delta + \boldsymbol\varepsilon_\ell$ with
residuals pairwise orthogonal, orthogonal to $\boldsymbol\delta$, and of
common energy $\varepsilon^2$, then under
Assumption~\ref{assump:indep} a uniform mixture over $L$ families satisfies
\[
    \mathbb{E}[\HHI] - \mathbb{E}[\HHI^{\mathrm{null}}]
    =
    \frac{M(M-1)\,\|\boldsymbol\delta\|_2^2}{K^2}
    +
    \frac{M\left(\tfrac{M}{L} - 1\right)\varepsilon^2}{K^2}.
\]
\end{restatement}

\begin{proof}
For two prompts in the same family,
$\langle \boldsymbol d_\ell, \boldsymbol d_\ell \rangle
= \|\boldsymbol\delta\|_2^2 + \varepsilon^2$ by the orthogonality of
$\boldsymbol\varepsilon_\ell$ to $\boldsymbol\delta$; for prompts in
different families $\ell \neq \ell'$,
$\langle \boldsymbol d_\ell, \boldsymbol d_{\ell'} \rangle = \|\boldsymbol\delta\|_2^2$ by
the pairwise orthogonality of the residuals. Summing over all
$\binom{M}{2}$ pairs, of which $L\binom{M/L}{2}$ are within-family,
\[
    2 \sum_{p<q} \langle d_p, d_q \rangle
    = M(M-1)\,\|\boldsymbol\delta\|_2^2
    + 2 L \tbinom{M/L}{2}\, \varepsilon^2
    = M(M-1)\,\|\boldsymbol\delta\|_2^2
    + M\left(\tfrac{M}{L} - 1\right) \varepsilon^2 ,
\]
and Lemma~\ref{lem:conc} under independence gives the claim. The first term
does not depend on $L$: it is the floor. Setting
$\boldsymbol\delta = 0$ recovers Corollary~\ref{cor:orth}.
\end{proof}

\subsection{Exclusion and the impossibility test}
\label{app:exclusion}

\begin{lemma}[The never-cited probability]
\label{lem:excl}
Under Assumption~\ref{assump:indep}, for any catalogue paper $j$,
\begin{align*}
    \Excl_j
    &\;=\;
    \Pr\!\left(j \text{ shown at least once but never cited}
    \mid j \text{ shown at least once}\right) \\[2pt]
    &\;=\;
    \frac{\prod_{p=1}^{M}\left(1 - s\, r_{pj}\right) - (1-s)^{M}}
         {1 - (1-s)^{M}} .
\end{align*}
\end{lemma}

\begin{proof}
Under independence, the event that $j$ is never cited has probability
$\prod_p (1 - q_{pj})$ with $q_{pj} = s\, r_{pj}$, and the event that $j$ is
never shown has probability $(1-s)^M$. Since $Y_{pj} \le X_{pj}$ almost
surely, ``never shown'' is contained in ``never cited,'' so
$\Pr(\text{shown at least once and never cited})
= \prod_p (1 - s r_{pj}) - (1-s)^M$. Conditioning on ``shown at least
once,'' of probability $1 - (1-s)^M$, gives the claim.
\end{proof}

\begin{restatement}{Proposition~\ref{prop:exch} (Exchangeability upper bound; restated)}
Under Assumption~\ref{assump:indep} and exchangeability within a group $A$
($r_{pj} = r^{A}_{p}$ for all $j \in A$, with group average $R_A$),
\[
    \Excl^{\mathrm{ex}}_{A}
    \;\le\;
    \frac{(1 - s R_A)^{M} - (1-s)^{M}}{1 - (1-s)^{M}},
\]
with equality iff all $r^{A}_{p}$ are equal.
\end{restatement}

\begin{proof}
Under exchangeability, Lemma~\ref{lem:excl} gives
$\Excl^{\mathrm{ex}}_A = \bigl[ \prod_p (1 - s\, r^A_p) - (1-s)^M \bigr] /
\bigl[ 1 - (1-s)^M \bigr]$. The function
$\varphi(r) = \log(1 - s r)$ is strictly concave on $[0, 1]$ since
$\varphi''(r) = -s^2/(1 - s r)^2 < 0$. By Jensen's inequality,
$\tfrac{1}{M} \sum_p \log(1 - s\, r^A_p)
\le \log\!\left(1 - s R_A\right)$, i.e.
$\prod_p (1 - s\, r^A_p) \le (1 - s R_A)^M$, with equality iff all
$r^A_p$ coincide. Substituting into the exact formula proves the bound.
\end{proof}

The impossibility reading used in Section~\ref{sec:collapse} is as
follows. At round 0, $s = 0.25$ and $M = 120$, so at the observed group
rates the bound is of order $10^{-5}$, and observed exclusion fractions of several percent are
infeasible under any exchangeable model, however heterogeneous across
prompts. Stable within-group heterogeneity, a paper-level map, is thereby
forced, without estimating it.

\subsection{Rarity amplification}
\label{app:rarity}

The multiplicative-lift model of Section~\ref{sec:recursion}: base
saliences $A_0, A_1, \dots$ are i.i.d.\ positive with nondegenerate law
$F$; a seed carries weight $b A_j$ with lift $b \ge 1$, a generated paper
carries $A_j$. For a marked shown seed of own base salience $a$, facing $x$
other seeds among its $m-1$ competitors, define
\begin{equation}
    g_x(a)
    \;=\;
    \mathbb{E}\!\left[
    \pi\!\left(b a;\; b A_1, \dots, b A_x,\; A_{x+1}, \dots, A_{m-1}\right)
    \right],
    \label{eq:gx}
\end{equation}
with $\pi$ the inclusion probability \eqref{eq:inclusion} at budget $k$,
and the expectation over the competitors' saliences.

\begin{restatement}{Proposition~\ref{prop:rarity} (Heterogeneous rarity amplification; restated)}
In the multiplicative-lift model, for every $a > 0$: $g_x(a)$ is strictly
decreasing in $x$ when $b > 1$, and does not depend on $x$ when $b = 1$. In
the latter case $\mathbb{E}_{A_0 \sim F}\,[g_x(A_0)] = k/m$ for every $x$;
the value at a fixed $a$ is not $k/m$ in general, since a marked seed of
unusually high salience is cited more often than one of unusually low
salience whatever the panel composition.
Conditional on a marked seed being shown at round $t$, its number of seed
competitors is
$X_t \sim \mathrm{Hypergeom}(N_t - 1,\; 119,\; m - 1)$, a family
stochastically decreasing in $t$; consequently the seed citation rate when
shown, $R_{\mathrm{seed}}(t) = \mathbb{E}\left[g_{X_t}(A_0)\right]$, which
averages over $A_0$, is increasing in $t$ when $b > 1$ and flat at $k/m$
when $b = 1$.
\end{restatement}

\begin{proof}
Fix a realized panel with budget $k$, a marked paper of weight $w$, and
competitor weights $v_1, \dots, v_{m-1}$; by \eqref{eq:inclusion},
$\pi = w\, e_{k-1}(v) / e_k(w, v)$. For one competitor weight $v_\ell$,
write $A$ for the remaining $m-2$ competitor weights; then
\[
    \frac{\pi}{1 - \pi}
    = w\, \frac{e_{k-1}(A) + v_\ell\, e_{k-2}(A)}{e_k(A) + v_\ell\, e_{k-1}(A)},
    \qquad
    \frac{\partial}{\partial v_\ell} \frac{\pi}{1-\pi}
    = w\, \frac{e_{k-2}(A)\, e_k(A) - e_{k-1}(A)^2}
               {\left(e_k(A) + v_\ell\, e_{k-1}(A)\right)^2}.
\]
By Newton's inequalities for elementary symmetric polynomials of positive
arguments, $e_{k-1}(A)^2 \ge e_{k-2}(A)\, e_k(A)$, so the derivative is
nonpositive: the marked paper's inclusion probability is weakly decreasing
in every competitor's weight. Passing from $g_x$ to $g_{x+1}$ replaces one
generated competitor of weight $A$ by a seed competitor of weight $bA$,
which weakly increases that competitor's weight pointwise; the decrease is
strict on a set of positive measure when $b > 1$ and $F$ is nondegenerate,
so $g_x(a)$ is strictly decreasing in $x$. When $b = 1$ all $m$ weights are
i.i.d., so the competitor law is the same whether a slot holds a seed or a
generated paper; the right-hand side of \eqref{eq:gx} therefore does not
depend on $x$, and we may write $g_x(a) = g(a)$. Exchangeability of the
marked paper with its competitors then gives
$\mathbb{E}_{A_0 \sim F}\,[g(A_0)] = k/m$, since the $m$ exchangeable
inclusion probabilities sum to $k$. Note that this identity holds only
after averaging over the marked paper's own salience: $g(a)$ itself is
increasing in $a$ and equals $k/m$ only at the salience where the average
is attained.

Conditional on the marked seed being shown at round $t$, the other $m - 1$
panel slots are a uniform draw from the remaining $N_t - 1$ papers, of
which $119$ are seeds, giving the hypergeometric law; as $N_t$ grows with
$m$ and the seed count fixed, the family is stochastically decreasing in
$t$. Because $g_x(a)$ is decreasing in $x$ for every $a$, and $A_0$ is
independent of $X_t$,
$R_{\mathrm{seed}}(t) = \mathbb{E}[g_{X_t}(A_0)]$ is increasing in $t$ for
$b > 1$. For $b = 1$ the same expectation is
$\mathbb{E}_{A_0}[g(A_0)] = k/m$ for every $t$, so the trajectory is flat.
\end{proof}

\begin{remark}[Scarcity boundary]
\label{rem:scarcity}
If $k_p = m$, every shown paper is cited and the choice law is
observationally indistinguishable from the null: the theory predicts the
uncapped pilot of Section~\ref{sec:benchmark} in the limit; at the pilot's
realized budgets of 22--27 the remark is qualitatively consistent rather than
exact. The budget is the
switch that makes preference visible; where between weak and strong
scarcity the transition happens is an empirical question, not a theoretical
one.
\end{remark}

\section{Benchmark Construction, Prompts, Parsing, and Compliance}
\label{app:benchmark}

\subsection{Corpus and Metadata}\label{app:corpus}

We built the benchmark on 120 real knowledge-distillation papers collected
from arXiv, published 2015--2022, each with between 50 and 500 citations at
collection time; the band excludes both obscure and canonical work so that
no paper is an outlier on visibility. Each paper enters the benchmark as a
blinded record with four fields: a fabricated single author surname, a
reassigned publication year (seed years drawn uniformly from 2017--2022;
generated papers during the recursion receive years in 2017--2025), the
paper's real title, and its real abstract. Citation counts, venues, real
author identities, and links are never shown. Fabricated surnames can
collide across records; a record's identity is therefore the (surname,
year) \emph{pair}, and pairs are unique within every panel. Domain labels
used by the probes of Section~\ref{sec:mapcontent} come from a separate
categorization of the real papers and are never shown to any selector. The
full corpus listing is App.~\ref{app:seeds}.

\subsection{Prompt Templates}\label{app:prompts}

Each prompt renders its panel of 30 blinded records into the following
template (the position variant; the review variant is byte-identical except
that the opening verb phrase ``argue a position on'' is replaced by
``synthesize what is known about''). A fixed system prompt, shared by all
vendors, instructs the model to return JSON with title, abstract, and body
fields and to cite inline as (Surname, Year):

\begin{quote}\footnotesize\ttfamily\raggedright
"Using only the articles provided, argue a position on the following
topic. Support your claims by citing relevant articles inline, written as
(Surname, Year). Cite only articles from the provided list, copying each
author surname and year exactly. Do not write "et al.", initials, square
brackets, or numbered citations. IMPORTANT: cite AT MOST 10 distinct
articles in total -- never more than 10. Topic: Knowledge distillation or
model compression in deep learning Articles: [\{'author': 'Gardner',
'year': 2020, 'title': 'Adversarially Robust Distillation', 'abstract':
'Knowledge distillation is effective for producing small, high-performance
neural networks for classification, but these small networks are
vulnerable to adversarial attacks. [\ldots]'\}, \textrm{\emph{followed by
the remaining 29 records of the panel in the same format}}]
\end{quote}

The topic line is fixed for the corpus; the panel records appear in the
prompt in their sampled order. Sixty prompts per round use the review task
sentence and sixty the position sentence.

\subsection{Parsing and Its Sensitivity}\label{app:parsing}

We parse citations deterministically from the generated text, through two
paths with different jobs. The authoritative path is table-driven: for each
prompt it searches the body only for the (surname, year) pairs of that
prompt's own panel, counting a pair as cited when the surname leads, an
optional possessive or ``et al.''/``and X'' tail follows, and at most eight
characters of citation punctuation (never letters or digits) separate it
from its exact year, with word boundaries on both sides. This rule accepts
every citation style the models produced, including styles the prompt
forbade, such as one model family's bracketed pair lists ``([Surname,
Year]; \ldots)'' and another's separately bracketed ``[Surname] [Year]'',
while structurally rejecting prose coincidences like ``In 2019, \ldots'',
where the year precedes the name or words intervene. Because the search
vocabulary is the panel itself, a false positive requires a shown surname
adjacent to its own displayed year, and misses are limited to paraphrases
with intervening words, which read as uncited; the parser ships with a
fourteen-style adversarial suite covering the observed conventions
(accents, particles, apostrophes, year-suffix letters, mixed bracket
styles, prose traps), which the table-driven path passes in full. A
bibliography is the \emph{set} of distinct resolved panel papers, so
repeated citations of one paper count once and $k_p$ is the number of
distinct valid citations. The second path, a style-tolerant regex with no
candidate table, is used only for hallucination accounting: any extracted
pair that matches no shown record, whether it names a catalogue paper not
in this panel or a paper that does not exist at all, is recorded as a
hallucination; the two cases are indistinguishable to the protocol and both
are removed before every analysis.

\subsection{Compliance and Data Quality}\label{app:compliance}

We report round-0 data quality per model in Table~\ref{tab:compliance}. Budgets
sit essentially at the cap (mean realized $k_p$ between 9.6 and 10.0), and
after regeneration of non-compliant runs no round-0 prompt exceeds the cap
for any model. Hallucination, measured as the mean per-prompt rate of
citations that resolve to no shown record, is at most 1.2\% for any model
at round 0. In the recursion it is at most 1.4\% for any model pooled over
rounds; per-round means reach 2.4\% for the weakest model in isolated
rounds, and every flagged citation is removed regardless, so no analysis
sees a hallucinated reference. The Claude models could not sustain cap
compliance across recursion rounds and are excluded from the recursion:
truncating an over-budget bibliography would censor the choice under study
rather than constrain it, and regeneration did not converge for these
models at recursion scale. The uncapped pilot referenced in
Section~\ref{sec:benchmark}, in which models cite 22--27 of 30 shown papers
and are statistically indistinguishable from the null, used the same
templates with the budget sentence removed.

\begin{table}[ht]
    \centering
    \caption{Round-0 compliance and data quality. Mean realized budget
    $\bar k$ and mean per-prompt hallucination rate; all hallucinated
    citations are removed before analysis.}
    \label{tab:compliance}
    \small
    \begin{tabular}{@{}lcc@{}}
        \toprule
        Model & Mean budget $\bar k$ & Hallucination rate (\%) \\
        \midrule
        GPT-4.1 mini & 9.86 & 1.12 \\
        Claude Haiku 4.5 & 10.00 & 0.36 \\
        Gemini 3.1 Flash-Lite & 9.76 & 0.34 \\
        GPT-5 mini & 9.86 & 0.22 \\
        Claude Opus 4.8 & 10.00 & 0.09 \\
        Gemini 2.5 Flash & 9.98 & 0.08 \\
        Claude Sonnet 4.6 & 10.00 & 0.00 \\
        Gemini 2.5 Pro & 9.99 & 0.00 \\
        Gemini 3.1 Pro & 9.91 & 0.00 \\
        GPT-4.1 & 9.62 & 0.00 \\
        GPT-5 & 10.00 & 0.00 \\
        \bottomrule
    \end{tabular}
\end{table}

\section{Maps, Reliability, and the Vote Histogram}
\label{app:maps}

\subsection{Map Estimation and Reliability}\label{app:mapest}

For each selector, we estimate a paper's citation rate when shown as
$\hat r_j = (\text{times cited}) / (\text{times shown})$ over the round's
prompts; each paper is shown roughly 30 times at round 0. We estimate the
null rate from the 150 uniform redraws that inherit every prompt's panel
and realized budget, and the map is the vector of differences
$\hat r_j - \hat r_j^{\mathrm{null}}$, following the definition of
Section~\ref{sec:benchmark}. We compute split-half reliability by randomly splitting the prompts in
half, estimating the map on each half, correlating the two estimates, and
applying the Spearman--Brown correction $2r/(1+r)$, averaged over 60 random
splits. We disattenuate cross-model correlations by dividing by
$\sqrt{\mathrm{rel}_a\,\mathrm{rel}_b}$, and the spectrum of
Figure~\ref{fig:onemap} is that of the cross-model covariance with each
diagonal entry reduced by the map's mean binomial sampling variance, so
that estimation noise does not masquerade as idiosyncratic preference.

\subsection{The Replicate Run}\label{app:replicate}

Three independent generations of GPT-5 mini's round 0 on identical panels
and protocol, differing only in sampling seed, reproduce the map at a mean
pairwise correlation of 0.973 (range 0.972--0.976) at the benchmark's
120-prompt design, and 0.993 at 480 prompts. This is the generation-noise
ceiling against which the other reliability figures in the paper are read,
and it sits above the split-half estimate of 0.955, as it must, since
split-half partitions prompts and so varies panels and list order as well as
generation. Per-prompt bibliographies overlap the original run at a mean
Jaccard of 0.52. The replicate is the resampling placebo used in
Section~\ref{sec:onemap}. A resampled model is not a new preference
family, so the share of excess that mixing removes cannot be manufactured
by sampling noise.

\subsection{The Vote Histogram, Predicted}\label{app:votehist}

Because all eleven models judge byte-identical panels, every (prompt, shown
paper) instance carries a vote count: the number of models citing that
paper in that prompt. Figure~\ref{fig:votes} shows the distribution of
these counts. The shared-map prediction treats the eleven decisions on one
instance as independent Bernoulli draws with the models' own estimated
rates $\hat r_{m j}$, a Poisson--binomial mixture aggregated over
instances; the independent-selectors baseline replaces every rate by the
common mean, a single binomial with matched marginals. The observed
distribution piles mass at both ends, 19\% of instances draw no citation,
46\% at most two, and 2.5\% are unanimous (probability $5 \times 10^{-6}$
under the matched binomial), and the shared-map prediction reproduces it at
total-variation distance 0.11 against 0.43 for the baseline.

\begin{figure}[ht]
    \centering
    \includegraphics[width=0.55\textwidth]{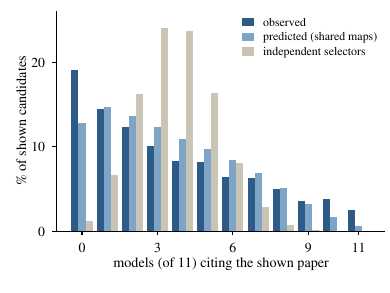}
    \caption{\textbf{Jointly decided.} Votes over shown candidates:
    observed; predicted from the estimated maps under conditional
    independence across models (total-variation distance 0.11);
    independent selectors with the same marginals (0.43).}
    \label{fig:votes}
\end{figure}

\section{The Noise-Corrected Alignment Ratio and Stable Rank}
\label{app:rho}

The estimator behind Table~\ref{tab:round0}'s $\rho$ column and the identity
check of Section~\ref{sec:onemap}. For each model, stack the realized-panel
deviation rows $x_p \in \mathbb{R}^{N_t}$, with $\hat r_{pj} - k_p/m$ on the
shown panel and zero elsewhere (the natural estimate of the deviation
vectors $d_p$ of App.~\ref{app:map}, up to the common exposure factor), and
form the prompt-level Gram matrix $G = X X^{\!\top} \in
\mathbb{R}^{M \times M}$. Its off-diagonal entries are unbiased for the
population inner products $\langle d_p, d_q \rangle$, since
$\mathbb{E}[x_{pj}] = s\,(r_{pj} - k_p/m) = \delta_{pj}$, but each diagonal
entry carries the prompt's own sampling variance in addition to
$\|d_p\|_2^2$; we therefore subtract from every diagonal entry the
plug-in inflation
$\hat v = \sum_j \left[ s\, \hat r_j (1 - \hat r_j)
+ s(1-s)(\hat r_j - \bar k / m)^2 \right]$,
yielding the corrected Gram $G_c$. The estimators are
\[
    \hat\rho
    = \frac{\mathbf{1}^{\!\top} G_c\, \mathbf{1}}{M \operatorname{tr} G_c},
    \qquad
    \hat r_s
    = \frac{\operatorname{tr} G_c}{\lambda_{\max}(G_c)} ,
\]
the sample versions of the alignment ratio and of the stable rank of
Proposition~\ref{prop:map}. We call $\hat r_s$ a corrected spectral ratio
rather than a stable rank, since diagonal subtraction can leave $G_c$
indefinite and the bound of Proposition~\ref{prop:map} then does not apply
to it. Across
the eleven models, $\hat\rho$ runs 0.89--0.98 against $1/M = 0.008$ for
unaligned prompts, and $\hat r_s$ runs 0.85--1.06, statistically
indistinguishable from one. Two of the eleven stable-rank point estimates fall marginally below one. A stable rank cannot be less than one for a positive semidefinite matrix, so those values indicate that the plug-in correction removes slightly more than the sampling noise, for the reason given above: the pooled per-paper rate absorbs between-prompt variation as well as within-prompt noise. We therefore read the corrected quantities as upper bounds and report the uncorrected values alongside them throughout.

The identity check uses the same objects with no correction needed,
because by Lemma~\ref{lem:conc} under Assumption~\ref{assump:indep} the
excess HHI equals $\bigl( \mathbf{1}^{\!\top} G \mathbf{1}
- \operatorname{tr} G \bigr) / K^2$, which involves only off-diagonal
entries. Figure~\ref{fig:identity} plots this reconstruction against the
observed excess (model HHI minus the mean of 200 uniform redraws inheriting
each prompt's panel and budget) for all eleven models: correlation 0.999,
mean relative deviation 1.6\%, within the redraw noise of the null itself.

\begin{figure}[ht]
    \centering
    \includegraphics[width=0.42\textwidth]{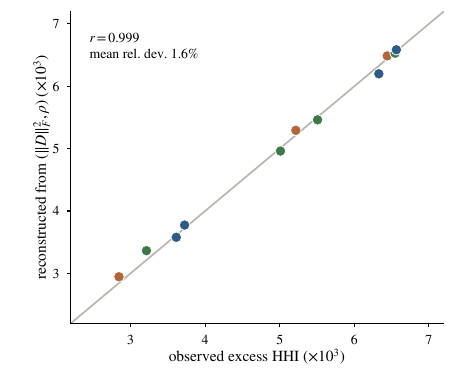}
    \caption{\textbf{The identity audits the pipeline.} Excess HHI
    reconstructed from the deviation Gram matrix via
    Eq.~\eqref{eq:sharedmap} against the observed excess over the matched
    null, one point per model, colored by vendor; the line is $y = x$.}
    \label{fig:identity}
\end{figure}

\section{The Human Study: Protocol, the Symmetry Computed, Robustness}
\label{app:human}

\subsection{Protocol}\label{app:protocol}

Eight researchers familiar with the knowledge-distillation literature
completed prompts through a web interface. Each session presented one
prompt's panel exactly as the models saw it: the same 30 blinded records
(fabricated surname, reassigned year, real title and abstract), the same
task sentence, and the same instruction to cite at most 10. The interface did
not hard-block a longer list, and one annotator selected 11 on one prompt;
every null inherits the realized budget, so the protocol is satisfied
prompt by prompt regardless. Annotators selected
the papers they would cite; free text was not required, so the human data is
selection-only. Because the models in the main protocol write prose, we also
reran all eleven in a matched selection-only mode
(App.~\ref{app:selonly}).
Between them the experts covered 53 of the 120 prompts, with effort
concentrated in one principal annotator (41 of 53); the imbalance is
addressed in \S\ref{app:robustness}. All human-model comparisons in the
paper are computed on the covered slice with model statistics recomputed on
the identical slice.

\subsection{The Symmetry, Computed}\label{app:symmetry}

The pooled human test of Section~\ref{sec:collapse} is the exchangeability
goodness-of-fit at the humans' own coverage. For each paper $j$, let $n_j$
be its shows and $c_j$ its citations across the covered prompts; under
exchangeability with per-prompt generosity $k_p / m$, paper $j$'s expected
citations are $e_j = \sum_{p \ni j} k_p / m$ with variance
$v_j = \sum_{p \ni j} (k_p/m)(1 - k_p/m)$, and the statistic is
$\chi^2 = \sum_j (c_j - e_j)^2 / v_j$ over shown papers, with degrees of
freedom the number of shown papers minus one. The pooled experts give
$\chi^2 = 121.5$ on 119 degrees of freedom, $p = 0.42$. The models'
verdicts on the same axes use Proposition~\ref{prop:exch} at the round-0
protocol ($s = 0.25$, $M = 120$). With each model's own observed group rate
$R$, the exchangeable exclusion probability per paper is of order
$10^{-5}$. The observed exclusion counts are priced by a panel-preserving
randomization rather than a binomial tail, since papers within a prompt
compete for a single budget and are therefore not independent: each realized
panel is redrawn under its own budget, 4{,}000 times. The null count never
exceeds one across those draws, while the models range from 2 to 20 and the
pooled experts give 2 against a null mean of 1.6 ($p = 0.48$). Monte Carlo
resolution bounds every model's $p$ at $2.5 \times 10^{-4}$; the separation
is one of support rather than of tail mass. The bound in
Figure~\ref{fig:symmetry} is drawn for the model protocol; the human point
is tested at its own coverage as above.

\subsection{Bounding Expert Preference Strength}\label{app:equiv}

Non-rejection of the uniform null is not evidence that the experts carry no
map, so we report an upper bound on how much map they could be carrying
undetected. This appendix gives the estimand, the estimator, and the
construction of the interval quoted in the abstract and in
Section~\ref{sec:collapse}.

\paragraph{Estimand.} For a selector on a fixed set of prompts, let $d_j$ be
paper $j$'s true citation propensity minus its matched-null rate, and define
\emph{preference strength} as $\mathcal{E} = \operatorname{Var}_j(d_j)$, the
dispersion of the map across papers. A selector with no map has
$\mathcal{E} = 0$ regardless of how generous or stingy it is overall, since
the null rate absorbs the level.

\paragraph{Estimator.} Let $\hat d_j$ be the observed rate minus the matched
null. Its variance across papers is inflated by sampling: with $n_j$
exposures per paper and a realized budget, $\operatorname{Var}_j(\hat d_j) =
\mathcal{E} + \nu$, where $\nu$ is the dispersion a mapless selector would
show on the same panels. We estimate $\nu$ by panel-preserving
randomization. Each realized panel is redrawn under its own realized budget,
sampling without replacement, which reproduces both the exposure pattern and
the negative within-prompt dependence that a fixed budget induces. Writing
$\nu_b$ for the map dispersion of draw $b$ over $B = 4{,}000$ draws,
\[
    \widehat{\mathcal{E}} \;=\; \operatorname{Var}_j(\hat d_j) \;-\;
    \tfrac{1}{B}\textstyle\sum_b \nu_b ,
    \qquad
    \widehat{\mathcal{E}}^{\,+} \;=\; \operatorname{Var}_j(\hat d_j) \;-\;
    q_{2.5}(\nu) ,
\]
with $q_{2.5}$ the $2.5$th percentile of the null dispersions. We call
$\widehat{\mathcal{E}}^{\,+}$ a protocol-matched sensitivity bound rather
than a formal $95\%$ upper confidence limit. It asks how large the true
preference strength could be while the observed dispersion stays consistent
with a mapless selector on these panels, but its coverage under a nonzero
heterogeneous map is not established, and establishing it would require
inverting a simulation over the actual panels and budgets.

\paragraph{Matching.} The experts completed 53 prompts, so every model is
restricted to the same 53 prompts, the same panels and the same realized
budgets before the comparison. This matters: on 53 prompts a paper receives
about 13 exposures rather than 30, so $\nu$ is roughly twice its full-run
value and every estimate is correspondingly looser. Comparing an expert bound
computed at 53 prompts against model estimates computed at 120 would
overstate the gap.

\paragraph{Result.} Matched in this way, the eleven models give
$\widehat{\mathcal{E}}$ from 0.033 to 0.082 with a mean of 0.058. The pooled
experts give $\widehat{\mathcal{E}} = 0.0006$ and
$\widehat{\mathcal{E}}^{\,+} = 0.0044$. The bound is therefore $7.5\%$
of the average model's estimate ($0.0044/0.0584$), and the least concentrated
single model exceeds the limit by a factor of $7.6$ ($0.0333/0.0044$). The two
ratios have different denominators and are not reciprocals of one another.

\paragraph{What this does and does not license.} It licenses the statement
that expert preference strength, if present, is at most a small fraction of
the models', at the coverage this study achieved. It does not license the
statement that the experts carry no map. Nor does it speak to individual
annotators: the estimand is a property of the pooled population, and pooling
heterogeneous individual maps reduces dispersion mechanically, so a pooled
bound is not a per-annotator bound. With 41 of the 53 prompts completed by a
single annotator the pooled quantity is close to that annotator's, but the
design cannot separate the two.

\subsection{Secondary Outcomes}\label{app:secondary}

Overlap: for each covered prompt and each frontier model, the number of
common citations between the human and model bibliographies is compared
with the hypergeometric law implied by the two budgets on the shared panel;
pooled over prompts, observed overlaps of 2.96--3.25 citations sit at the
expected $\approx\!3.06$ (Fisher-combined $p = 0.90$--$1.00$), while the
same models overlap one another at roughly twice chance on the same slice.
Stability: the split-half reliability of the pooled human map is
$0.03 \pm 0.13$ over random prompt splits (App.~\ref{app:mapest}
procedure), against 0.86--0.95 for every model. Spectrum: the raw stable
rank of the human prompt-level deviation matrix is 21, above every model's
(8.8--16.3), and after the noise correction of App.~\ref{app:rho} the
surviving signal energy is under 1\% of raw, against 5--11\% for the
models. Orthogonality is the supporting estimate. The pooled human map correlates
$r = 0.19$ with the consensus map, a projection coefficient of 0.12 on the
consensus direction and under 4\% shared variance
(Figure~\ref{fig:humanapp}a); with human reliability this low, the
disattenuated value is undefined, which is why we treat orthogonality as
directional support rather than a claim the argument depends on.

\subsection{Robustness and Scope}\label{app:robustness}

The human claims that carry weight are population-level and survive the
annotator imbalance. We compute the identical statistics for three
populations: all eight experts, the principal annotator alone, and the
seven others without the principal. Top-decile shares are 18.9\%, 19.6\%,
and 31.4\% against matched nulls of 18.7\%, 19.9\%, and 29.8\% (the
small-sample null is mechanically higher), with exchangeability $p$-values
of 0.22, 0.31, and 0.23; every human population sits at its own null and
passes the same test every model fails (Figure~\ref{fig:humanapp}b). One bookkeeping note: the $p$-values in this
table use the plug-in variance convention above, applied identically across
rows; the main text reports the analysis pipeline's value for the pooled
population ($\chi^2 = 121.5$, $p = 0.42$), which differs slightly in the
variance convention; the verdict is the same under both. Scope: 53 prompts
bound the precision of any human map estimate, which is why absence claims
run through the exchangeability test rather than through estimation, and
blinded selection is not citation practice in the wild.

\begin{figure}[ht]
    \centering
    \includegraphics[width=0.86\textwidth]{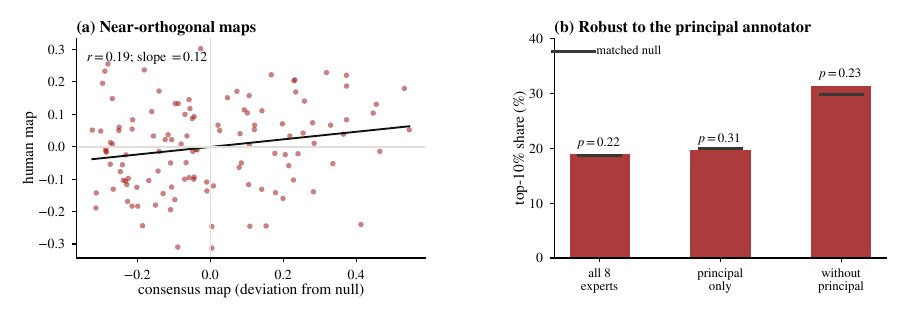}
    \caption{\textbf{The human study's supporting artifacts.}
    (a)~The pooled human map against the consensus map: $r = 0.19$,
    projection coefficient 0.12. (b)~Top-decile share for three human
    populations against their matched nulls (dashes), with exchangeability
    $p$-values: all at the null, including the principal annotator alone
    and the population without them.}
    \label{fig:humanapp}
\end{figure}

\section{What the Map Is Made Of: Probes and Their Caveats}
\label{app:mapchar}

The central claim of Section~\ref{sec:mapcontent} is the
cross-validated text model. The analysis pipeline's pre-specified model
predicts a held-out paper's consensus deviation from title and abstract at
$r = 0.47$ out of sample, the value reported in the main text. An
independent reimplementation here reproduces it. TF--IDF features
(unigrams and bigrams, English stop words removed, sublinear term
frequency, minimum document frequency 3) with ridge regression
($\alpha = 1$) under 10-fold cross-validation reach out-of-sample
$r = 0.49$ (Figure~\ref{fig:textmodel}); nearby configurations in a small
grid give similar values, and we note the mild optimism inherent in
reporting a grid's best configuration, which is why the main text cites the
pipeline's pre-specified 0.47 rather than this reproduction. Since the
displayed author names and years are fabricated, the predictable component
is a function of what the abstract says.

The probes are post hoc and should be read as such. The
methodological-coreness index is computed from each paper's title plus
abstract, lowercased. For each of two keyword lists we count raw substring
occurrences, so that ``distill'' also matches ``distillation'' and ``soft''
also matches ``softmax'', and repeated mentions count repeatedly; dividing
by the record's word count gives occurrences per 100 words, and the index
is the core-methods density minus the applications density. The
core-methods list is (teacher, student, soft, logit, temperature, dark
knowledge, mimic, distill, feature, layer, hint, representation); the
applications list is (video, 3d, speech, federated, reinforcement, policy,
agent, translation, recommend, medical, graph, detection, segmentation,
cross-modal, multimodal). The index runs from $-8.0$ to $+19.3$ across the
corpus; its extremes behave as intended, with ``Distilling Knowledge by
Mimicking Features'' at the top and a 3D object-detection paper at the
bottom. It correlates 0.54 with the consensus map and 0.26 with the human
map on matched axes. Domain labels are the corpus annotation of
Table~\ref{tab:seeds}: each seed carries one primary application-domain
label, assigned when the corpus was built (CV 51, NLP 30, Theory 11,
Graph 6, Federated 6, Speech 5, and four categories with fewer than five
papers), entered as one-hot dummies. The $R^2$ values are in-sample
ordinary least squares with an intercept, descriptive rather than
predictive and therefore not comparable to the cross-validated text model
above: coreness alone explains $R^2 = 0.29$ of the consensus map, domain
dummies alone 0.17, and both together 0.40. For the domain-tilt panel of
Figure~\ref{fig:mapcontent}c we regress the consensus map on coreness
alone, take mean residuals within each domain having at least five papers,
and Theory and Graph retain positive mean residuals and Speech a negative one. The
keyword lists were chosen by the authors, the index is one of many possible
operationalizations, and none of our claims depend on it: if the
coreness story were wrong, the text model, the shared map, and the human
asymmetry would stand untouched.

\begin{figure}[ht]
    \centering
    \includegraphics[width=0.42\textwidth]{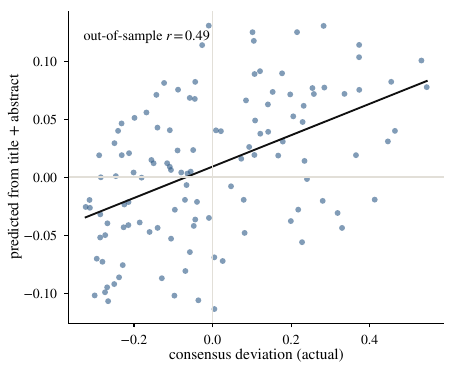}
    \caption{\textbf{The map is predictable from text.} Out-of-sample
    predictions of a TF--IDF ridge model on title and abstract against the
    actual consensus deviation, 10-fold cross-validation, one point per
    seed paper.}
    \label{fig:textmodel}
\end{figure}

\section{The Paraphrase Control}
\label{app:paraphrase}

The control tests whether the map depends on the seeds' surface form,
which any model may have memorized, or on their content. Each of the
eleven registry models rewrote the title and abstract of all 120 seed
records, preserving record identity and order so that panels, prompts, and
the RNG stream are byte-identical to the original run; GPT-5 mini then
replayed round 0 against each of the eleven paraphrased corpora
(1{,}320 prompts in total). A valid paraphrase must satisfy two
requirements that pull against each other, and failing either invalidates
the arm: the surface must be destroyed (mean 5-gram overlap with the
original at most 0.05, or a surviving memorization cue makes a null result
meaningless) and the meaning must be preserved (coreness-index correlation
above 0.95 with length change within 15\%, or the content variable itself
has moved and any difference is confounded). Four paraphrasers pass
cleanly, three graze one threshold each, and four are excluded: two kept
roughly three times the permitted surface (overlap 0.137 and 0.148) and
two summarized rather than paraphrased (coreness $r = 0.933$, length
$-15$\% and $-19$\%).

The map survives. Compliance is unchanged across every arm (mean budgets
9.80--9.97, hallucination 0.00--0.33\%), the collapse persists at its
original size (top-decile share 27.9--28.9\% against 28.5 originally, with
matched nulls near 15.6\%; 14--20 papers never cited against 15), and the
disattenuated correlation between each arm's map and the original-seed map
is 0.96--1.00 for the clean arms and 0.92--1.00 overall
(Figure~\ref{fig:paraphrase}). The four clean arms, written by four
different models from three vendors, correlate 0.95--0.99 with one another
(mean 0.97), so the signal
is the seeds' content rather than any paraphraser's style, and each arm's
map correlates 0.52--0.59 with the coreness index computed on the original
text, against 0.50 for the original map. The residual dose-response is
itself informative: extrapolated to zero surface overlap the map
correlation is 0.96, so surface cues contribute at most a few points, and
the only arms that move the map are the two whose paraphrases shifted
meaning, a positive control showing the map responds to content
perturbation exactly as a content readout should. Table~\ref{tab:paraqc}
reports every arm. The control runs one
model; paraphrase alone cannot exclude semantic recognition of a known
paper, but that residual channel must itself operate through content, and
the generated-only excess of Section~\ref{sec:recursion} shows the
collapse on papers that cannot appear in any training corpus.

\begin{table}[ht]
\centering
\caption{\textbf{The paraphrase control, arm by arm.} Quality control and
outcome for each paraphrasing model, sorted by residual surface overlap.
Bold marks a violated criterion; the last row is the disattenuated
correlation between GPT-5 mini's map on that rewritten corpus and its
original-seed map.}
\label{tab:paraqc}
\scriptsize
\setlength{\tabcolsep}{2.4pt}
\begin{tabular}{@{}p{2.9cm} ccccccccccc@{}}
\toprule
& \rotatebox{55}{Gemini 3.1 Pro}
& \rotatebox{55}{\shortstack[l]{Gemini 3.1\\Flash-Lite}}
& \rotatebox{55}{\shortstack[l]{Claude\\Haiku 4.5}}
& \rotatebox{55}{GPT-4.1}
& \rotatebox{55}{GPT-4.1 mini}
& \rotatebox{55}{\shortstack[l]{Gemini 2.5\\Flash}}
& \rotatebox{55}{\shortstack[l]{Claude\\Opus 4.8}}
& \rotatebox{55}{Gemini 2.5 Pro}
& \rotatebox{55}{\shortstack[l]{Claude\\Sonnet 4.6}}
& \rotatebox{55}{GPT-5}
& \rotatebox{55}{GPT-5 mini} \\
\midrule
mean 5-gram overlap ($\le 0.05$)
& .010 & .012 & .013 & .015 & .023 & .040 & .045 & .050 & \textbf{.056} & \textbf{.137} & \textbf{.148} \\
records with any reused 5-gram (of 120)
& 72 & 68 & 74 & 84 & 100 & 112 & 110 & 119 & 118 & 120 & 120 \\
worst-record overlap
& .098 & .078 & .067 & .135 & .135 & .271 & .250 & .182 & .266 & .347 & .391 \\
coreness $r$ ($> 0.95$)
& .950 & \textbf{.933} & \textbf{.933} & .959 & .958 & .962 & .984 & .976 & .973 & .981 & .975 \\
length change ($\pm 15\%$)
& $-2.4$ & $\mathbf{-18.7}$ & $\mathbf{-15.2}$ & $+0.6$ & $-7.8$ & $-1.8$ & $+3.2$ & $-4.1$ & $+0.3$ & $-7.5$ & $-3.0$ \\
\midrule
map correlation (disattenuated)
& .95 & .92 & .93 & .96 & .96 & .98 & 1.00 & .97 & .97 & 1.00 & .99 \\
\bottomrule
\end{tabular}
\end{table}

\section{Separating the Factors Bound to a Record}
\label{app:design}

Paraphrase varies the surface wording and holds the sampling design
byte-identical, so it cannot speak to the other properties bound to a
record. Table~\ref{tab:factors} lists the experiments that separate them.
All run GPT-5 mini at round 0 unless noted, and all measure a map against a
null matched to the same panels and budget.

\begin{table}[h]
\centering
\caption{Experiments separating the factors bound to a record. ``Held''
means byte-identical to the reference run. The crossover moves the entire
text to a different record while everything else stays with the slot.}
\label{tab:factors}
\small
\begin{tabular}{llll}
\toprule
experiment & varies & holds & measures \\
\midrule
paraphrase & wording & meaning, metadata, panel, order & surface form \\
crossover & text location & metadata, panel, order & content vs.\ slot \\
resampling & metadata, panel, order, task & text & everything but content \\
\texttt{meta\_only} & metadata & text, panel, order & surname and year \\
\texttt{order\_only} & order & text, metadata, panel & list position \\
\bottomrule
\end{tabular}
\end{table}

\paragraph{Crossover.} Six derangements $\pi_a$ over the 120 records; in arm
$a$, slot $i$ receives the title and abstract of record $\pi_a(i)$ and keeps
its own surname, year, panel memberships and positions. Arms run at the
reference seeds, so the design attaches to the slot automatically. Before
any arm was issued we asserted that panel membership, within-panel order and
every assigned $(\text{surname}, \text{year})$ pair were byte-identical to
the reference run across all prompts, with title and abstract the only text
differing: 3{,}600 cells changed at the 120-prompt design and 14{,}400 at
480, with zero non-text differences. Compliance matches the reference run in
every arm, at 9.86--9.95 distinct papers cited against a cap of ten.
Table~\ref{tab:xover} reports the six arms at the benchmark design.

\begin{table}[h]
\centering
\caption{Content-identity crossover, GPT-5 mini at the 120-prompt design.
Content share is $\beta_{\mathrm{content}}$ over the coefficient sum;
intervals are prompt bootstrap.}
\label{tab:xover}
\begin{tabular}{lccccc}
\toprule
arm & $\beta_{\mathrm{content}}$ & $\beta_{\mathrm{slot}}$ & sum &
content share & 95\% CI \\
\midrule
$\pi_1$ & 0.904 & \phantom{$-$}0.033 & 0.937 & 0.964 & 0.934--0.998 \\
$\pi_2$ & 0.916 & \phantom{$-$}0.029 & 0.945 & 0.969 & 0.933--1.014 \\
$\pi_3$ & 0.881 & \phantom{$-$}0.062 & 0.944 & 0.934 & 0.893--0.978 \\
$\pi_4$ & 0.922 & $-$0.019 & 0.902 & 1.021 & 0.986--1.061 \\
$\pi_5$ & 0.897 & \phantom{$-$}0.075 & 0.972 & 0.923 & 0.890--0.956 \\
$\pi_6$ & 0.911 & \phantom{$-$}0.020 & 0.931 & 0.978 & 0.943--1.016 \\
\midrule
mean & 0.905 & \phantom{$-$}0.033 & 0.939 & 0.964 & 0.928--1.002 \\
\bottomrule
\end{tabular}
\end{table}

Because the two regressors are the same vector under different index
orderings, they carry equal variance and are near-orthogonal for a random
derangement; observed regressor correlations run $-0.12$ to $+0.06$ with
variance inflation factors of 1.00--1.02. Under an additive
content-plus-slot model the coefficients are variance shares scaled by the
reliability of $d^0$, so their sum estimates that reliability rather than
testing against one. At 480 prompts the sum rises from 0.939 to 0.964 as
predictor error falls, and the content share is stable arm by arm across the
two designs. No individual arm's slot coefficient differs from zero under a
record bootstrap, and the spread across arms, including the slightly
negative value at $\pi_4$, sits within sampling error; we report the arms
unadjusted.

One scope note. Because the permuted corpus also changes what a record's
competitors are about, competitor-driven advantage is re-randomized between
the reference run and the crossover rather than held fixed, so it enters the
residual and not $\beta_{\mathrm{slot}}$. The crossover bounds the metadata
and position factors; panel composition is what design resampling covers.

\paragraph{Design resampling.} The single generator that had driven panel
membership, within-panel order, metadata assignment and the task split was
separated into four independent streams, defaulting to values that reproduce
the original run's prompt construction exactly. Six replicates then move all
four to fresh values with the corpus untouched. Each replicate's null is
recomputed against its own panels. Realized exposure is uniform at the
480-prompt design (mean 120.0, s.d.\ 8.4, range 99--144).

\paragraph{Single-factor arms.} Each redraws exactly one stream and holds
the others at reference values, so the distance from the generation ceiling
is a direct measurement of one factor with nothing inferred by subtraction.
An earlier attribution of this residual to list order, obtained by
subtracting a linear year term from the split-half comparison, was wrong:
split-half partitions prompts and therefore already varies order, so that
comparison isolates metadata rather than order. The direct arms reverse the
attribution, and the year term recovers only a fifth of the metadata factor
because surname identity is unmodelled by it.

\paragraph{Cross-model replication.} The crossover (three derangements) and
the metadata arm were repeated on GPT-4.1 mini at 480 prompts against
byte-identical designs, asserted on the issued prompts rather than assumed.
That model's generation ceiling is 0.976 at 480, against 0.993 for GPT-5
mini, a factor of 3.4 in variance-share units, so every gap is converted to
$(1-r)/r$ before any cross-model comparison and each model is read against
its own ceiling. The resolution threshold for declaring the two metadata
channels different, $0.015$, was fixed before the arms were run.

\section{The Selection-Only Arm}
\label{app:selonly}

The experts selected papers; the models in the main protocol wrote reviews
with citations embedded. To remove that asymmetry we reran all eleven models
on the 53 expert-covered prompts, asking each to return the papers it would
cite as a list of $(\text{Surname}, \text{Year})$ strings and nothing else.
Panels, blinding, task sentence and the cap of 10 are unchanged, and because
the citation format is unchanged the parser is the one used everywhere else.
Realized budgets stay at the cap, from 9.96 to 10.00 across models, against
2--11 for the experts, so each null inherits its own selector's budget as
before.

\begin{table}[h]
\centering\small
\caption{Selection-only mode on the 53 expert-covered prompts. Nulls are
matched per selector. Review-mode values on the same slice are 24.5--31.9\%
top-decile share and 8--26 exclusions.}
\label{tab:selonly}
\begin{tabular}{lrrrr}
\toprule
selector & top-10\% & null & uncited & null \\
\midrule
Claude Sonnet 4.6      & 34.5\% & 18.2\% & 51 & 1.1 \\
Gemini 2.5 Pro         & 34.5\% & 18.2\% & 46 & 1.0 \\
Claude Opus 4.8        & 33.8\% & 18.3\% & 53 & 1.1 \\
GPT-5                  & 31.9\% & 18.2\% & 53 & 1.1 \\
GPT-5 mini             & 31.3\% & 18.2\% & 35 & 1.1 \\
Gemini 2.5 Flash       & 30.9\% & 18.3\% & 29 & 1.0 \\
Claude Haiku 4.5       & 28.8\% & 18.3\% & 22 & 1.1 \\
Gemini 3.1 Pro         & 28.2\% & 18.2\% & 15 & 1.1 \\
GPT-4.1                & 27.4\% & 18.2\% & 14 & 1.0 \\
Gemini 3.1 Flash-Lite  & 24.0\% & 18.3\% & 16 & 1.0 \\
GPT-4.1 mini           & 22.9\% & 18.3\% &  6 & 1.0 \\
\midrule
pooled experts         & 18.9\% & 18.7\% &  2 & 1.5 \\
\bottomrule
\end{tabular}
\end{table}

Concentration survives the format change in every model, and the exclusion
counts rise rather than fall, reaching 53 of the 120 papers where review mode
on the same slice reached 26. The human-model gap therefore does not come
from the models having had to compose an argument.

\section{The Within-Panel Competition Test}
\label{app:regression}

Table~\ref{tab:regression} reports the full specification of
Eq.~\eqref{eq:competition} for all eight recursion models. Rows are (shown
seed, prompt) pairs over rounds 1--11 (7{,}557 rows in 1{,}292 prompt
clusters per model); the regressor is the seed's number of co-shown seeds;
round fixed effects absorb everything that drifts across rounds, seed fixed
effects absorb paper quality, a review-task dummy absorbs the prompt-type
split, and standard errors are clustered by prompt. The seed-side
coefficients are the analysis pipeline's; an independent reimplementation
for this appendix (Frisch--Waugh with alternating demeaning) reproduces
every coefficient and standard error to within 0.05\,pp.

Two companion columns complete the design. The \emph{placebo} runs the
identical regression with the outcome replaced by the citation frequency of
the same seed in the same prompt across 40 uniform redraws inheriting the
panel and budget; under the null the coefficient is exactly zero in
expectation, and the estimates sit within $\pm 0.08$\,pp for every model
(single redraws instead of the 40-redraw average scatter with the expected
sampling noise of one draw, about $\pm 0.5$\,pp, which is why the placebo is
defined on the average). The \emph{generated-side} regression replaces the
rows with (shown generated paper, prompt) pairs (32{,}043 rows), keeps the
same regressor, and uses generated-paper fixed effects. Its slopes are
smaller than the seed-side slopes for the strongly lifted models, by factors
of 3.3 to 4.4 for Gemini 3.1 Pro, GPT-5 and Gemini 3.1 Flash-Lite and
eighteen-fold for GPT-5 mini, and are of mixed sign ($-0.99$ to $+0.46$),
with the two weakest-lift models showing positive coefficients. The
exception is Gemini 2.5 Flash, whose two slopes are nearly equal ($-1.05$
against $-0.99$). The seed-side pattern is therefore not in general a
generic crowding artifact of panel composition, though for that one model
the two cannot be separated. Round 0 is excluded throughout because its panels contain 30
seeds and the regressor has no variation.

\begin{table}[ht]
    \centering
    \caption{The within-panel competition test, full table. Seed-side
    $\beta$ (pp per additional co-shown seed) with prompt-clustered
    standard errors and $t$; the 40-redraw placebo; and the generated-side
    slope. Sorted by seed-side $\beta$.}
    \label{tab:regression}
    \small
    \begin{tabular}{@{}lrrrrrr@{}}
        \toprule
        Model & $\beta_{\mathrm{seed}}$ & (SE) & $t$ &
        $\beta_{\mathrm{gen}}$ & (SE) & Placebo \\
        \midrule
        Gemini 3.1 Pro & $-3.07$ & $(0.27)$ & $-11.3$ & $-0.92$ & $(0.13)$ & $+0.035$ \\
        GPT-5 & $-2.57$ & $(0.21)$ & $-12.2$ & $-0.59$ & $(0.10)$ & $+0.031$ \\
        GPT-5 mini & $-1.80$ & $(0.22)$ & $-8.4$ & $-0.10$ & $(0.10)$ & $+0.073$ \\
        Gemini 3.1 Flash-Lite & $-1.73$ & $(0.20)$ & $-8.5$ & $-0.51$ & $(0.09)$ & $-0.023$ \\
        Gemini 2.5 Flash & $-1.05$ & $(0.25)$ & $-4.2$ & $-0.99$ & $(0.09)$ & $-0.015$ \\
        GPT-4.1 & $-0.94$ & $(0.26)$ & $-3.7$ & $0.37$ & $(0.11)$ & $+0.029$ \\
        GPT-4.1 mini & $-0.35$ & $(0.23)$ & $-1.5$ & $0.46$ & $(0.10)$ & $+0.017$ \\
        Gemini 2.5 Pro & $-0.31$ & $(0.23)$ & $-1.3$ & $-0.11$ & $(0.10)$ & $+0.018$ \\
        \bottomrule
    \end{tabular}
\end{table}

\section{The Seed Corpus}
\label{app:seeds}

The 120 real knowledge-distillation papers, listed with their primary
domain label as used by the probes of Section~\ref{sec:mapcontent}. Titles
are the papers' real titles (real titles and abstracts are displayed in the
benchmark; the blinded fields are authors, years, counts, and venues, per
App.~\ref{app:corpus}).

{\footnotesize
\begin{longtable}{@{}p{0.7cm}p{10.6cm}p{2.9cm}@{}}
    \caption{The seed corpus.}\label{tab:seeds}\\
    \toprule
    \# & Title & Domain \\
    \midrule
    \endfirsthead
    \toprule
    \# & Title & Domain \\
    \midrule
    \endhead
    \bottomrule
    \endfoot
        0 & Less is More: Task-aware Layer-wise Distillation for Language Model Compression & NLP \\
        1 & Mitigating Gender Bias in Distilled Language Models via Counterfactual Role Reversal & NLP \\
        2 & A Deep Hierarchical Approach to Lifelong Learning in Minecraft & RL/Policy \\
        3 & BERT Learns to Teach: Knowledge Distillation with Meta Learning & NLP \\
        4 & DE-RRD: A Knowledge Distillation Framework for Recommender System & RecSys \\
        5 & Lifelong Language Knowledge Distillation & NLP \\
        6 & Knowledge Distillation from Internal Representations & NLP \\
        7 & Knowledge Distillation for Improved Accuracy in Spoken Question Answering & Speech \\
        8 & MEAL: Multi-Model Ensemble via Adversarial Learning & CV \\
        9 & Learning Student-Friendly Teacher Networks for Knowledge Distillation & CV \\
        10 & Self-Knowledge Distillation in Natural Language Processing & NLP \\
        11 & Knowledge Distillation in Wide Neural Networks: Risk Bound, Data Efficiency and Imperfect Teacher & Theory \\
        12 & Highlight Every Step: Knowledge Distillation via Collaborative Teaching & CV \\
        13 & Channel Distillation: Channel-Wise Attention for Knowledge Distillation & CV \\
        14 & Adversarially Robust Distillation & CV \\
        15 & Towards Understanding Ensemble, Knowledge Distillation and Self-Distillation in Deep Learning & Theory \\
        16 & Distilling Knowledge From a Deep Pose Regressor Network & CV \\
        17 & Knowledge Distillation with Adversarial Samples Supporting Decision Boundary & CV \\
        18 & Explaining Knowledge Distillation by Quantifying the Knowledge & Theory \\
        19 & Selective Knowledge Distillation for Neural Machine Translation & NLP \\
        20 & Robust and Resource-Efficient Data-Free Knowledge Distillation by Generative Pseudo Replay & CV \\
        21 & Understanding Knowledge Distillation in Non-autoregressive Machine Translation & NLP \\
        22 & Show, Attend and Distill: Knowledge Distillation via Attention-based Feature Matching & CV \\
        23 & Reinforced Multi-Teacher Selection for Knowledge Distillation & NLP \\
        24 & Confidence-Aware Multi-Teacher Knowledge Distillation & CV \\
        25 & Feature-map-level Online Adversarial Knowledge Distillation & CV \\
        26 & Compressing Deep Graph Neural Networks via Adversarial Knowledge Distillation & Graph \\
        27 & Zero-Shot Knowledge Distillation in Deep Networks & CV \\
        28 & Fair Feature Distillation for Visual Recognition & CV \\
        29 & Distilling Policy Distillation & RL/Policy \\
        30 & Distilling Object Detectors with Fine-grained Feature Imitation & CV \\
        31 & Attention-Guided Answer Distillation for Machine Reading Comprehension & NLP \\
        32 & BERT-of-Theseus: Compressing BERT by Progressive Module Replacing & NLP \\
        33 & Does Knowledge Distillation Really Work? & Theory \\
        34 & Learning Generalizable Models for Vehicle Routing Problems via Knowledge Distillation & RL/Policy \\
        35 & Object DGCNN: 3D Object Detection using Dynamic Graphs & CV \\
        36 & Apprentice: Using Knowledge Distillation Techniques To Improve Low-Precision Network Accuracy & CV \\
        37 & Transferring Inductive Biases through Knowledge Distillation & Theory \\
        38 & Multi-Label Image Classification via Knowledge Distillation from Weakly-Supervised Detection & CV \\
        39 & Linkless Link Prediction via Relational Distillation & Graph \\
        40 & DINE: Domain Adaptation from Single and Multiple Black-box Predictors & CV \\
        41 & Improving the Interpretability of Deep Neural Networks with Knowledge Distillation & CV \\
        42 & Understanding BERT Rankers Under Distillation & NLP \\
        43 & FedZKT: Zero-Shot Knowledge Transfer towards Resource-Constrained Federated Learning with Heterogeneous On-Device Models & Federated \\
        44 & Compressing Visual-linguistic Model via Knowledge Distillation & Multimodal \\
        45 & Self-Knowledge Distillation with Progressive Refinement of Targets & CV \\
        46 & Moonshine: Distilling with Cheap Convolutions & CV \\
        47 & Model Compression with Two-stage Multi-teacher Knowledge Distillation for Web Question Answering System & NLP \\
        48 & What Makes a Good Data Augmentation in Knowledge Distillation -- A Statistical Perspective & Theory \\
        49 & Collaborative Distillation for Ultra-Resolution Universal Style Transfer & CV \\
        50 & A Closer Look at Deep Learning Heuristics: Learning rate restarts, Warmup and Distillation & Theory \\
        51 & Large scale distributed neural network training through online distillation & Theory \\
        52 & Embracing the Dark Knowledge: Domain Generalization Using Regularized Knowledge Distillation & CV \\
        53 & Revisiting Knowledge Distillation via Label Smoothing Regularization & Theory \\
        54 & Online Knowledge Distillation via Mutual Contrastive Learning for Visual Recognition & CV \\
        55 & Data-Free Adversarial Distillation & CV \\
        56 & Distilling the Knowledge of BERT for Sequence-to-Sequence ASR & Speech \\
        57 & Improved Feature Distillation via Projector Ensemble & CV \\
        58 & Rethinking Soft Labels for Knowledge Distillation: A Bias-Variance Tradeoff Perspective & Theory \\
        59 & TRILLsson: Distilled Universal Paralinguistic Speech Representations & Speech \\
        60 & A Fast Knowledge Distillation Framework for Visual Recognition & CV \\
        61 & ALP-KD: Attention-Based Layer Projection for Knowledge Distillation & NLP \\
        62 & Well-Read Students Learn Better: On the Importance of Pre-training Compact Models & NLP \\
        63 & Curriculum Learning for Dense Retrieval Distillation & NLP \\
        64 & BAM! Born-Again Multi-Task Networks for Natural Language Understanding & NLP \\
        65 & Private Model Compression via Knowledge Distillation & CV \\
        66 & Is Label Smoothing Truly Incompatible with Knowledge Distillation: An Empirical Study & Theory \\
        67 & Towards Model Agnostic Federated Learning Using Knowledge Distillation & Federated \\
        68 & PKD: General Distillation Framework for Object Detectors via Pearson Correlation Coefficient & CV \\
        69 & Knowledge Distillation from A Stronger Teacher & CV \\
        70 & Preserving Privacy in Federated Learning with Ensemble Cross-Domain Knowledge Distillation & Federated \\
        71 & Improving Neural Topic Models using Knowledge Distillation & NLP \\
        72 & Ensemble Knowledge Distillation for CTR Prediction & RecSys \\
        73 & Few Shot Network Compression via Cross Distillation & CV \\
        74 & On Representation Knowledge Distillation for Graph Neural Networks & Graph \\
        75 & Rejuvenating Low-Frequency Words: Making the Most of Parallel Data in Non-Autoregressive Translation & NLP \\
        76 & Distilled Semantics for Comprehensive Scene Understanding from Videos & CV \\
        77 & Federated Knowledge Distillation & Federated \\
        78 & Unifying Heterogeneous Classifiers with Distillation & CV \\
        79 & TextKD-GAN: Text Generation using Knowledge Distillation and Generative Adversarial Networks & NLP \\
        80 & Ranking Distillation: Learning Compact Ranking Models With High Performance for Recommender System & RecSys \\
        81 & Ensemble Distillation for Neural Machine Translation & NLP \\
        82 & Understanding and Improving Lexical Choice in Non-Autoregressive Translation & NLP \\
        83 & Towards Practical Lipreading with Distilled and Efficient Models & Speech \\
        84 & Knowledge Adaptation: Teaching to Adapt & NLP \\
        85 & Zero-shot Knowledge Transfer via Adversarial Belief Matching & CV \\
        86 & Extract the Knowledge of Graph Neural Networks and Go Beyond it: An Effective Knowledge Distillation Framework & Graph \\
        87 & Efficient Transformer-based Large Scale Language Representations using Hardware-friendly Block Structured Pruning & NLP \\
        88 & Distilling Task-Specific Knowledge from BERT into Simple Neural Networks & NLP \\
        89 & Graph-less Neural Networks: Teaching Old MLPs New Tricks via Distillation & Graph \\
        90 & Unified Visual Transformer Compression & CV \\
        91 & Robust Cross-Modal Representation Learning with Progressive Self-Distillation & Multimodal \\
        92 & AlphaNet: Improved Training of Supernets with Alpha-Divergence & CV \\
        93 & CoReD: Generalizing Fake Media Detection with Continual Representation using Distillation & CV \\
        94 & BERT-EMD: Many-to-Many Layer Mapping for BERT Compression with Earth Mover's Distance & NLP \\
        95 & Compression of Deep Learning Models for Text: A Survey & Survey \\
        96 & AutoGAN-Distiller: Searching to Compress Generative Adversarial Networks & CV \\
        97 & VPN++: Rethinking Video-Pose embeddings for understanding Activities of Daily Living & CV \\
        98 & Distilling Knowledge by Mimicking Features & CV \\
        99 & Compressing GANs using Knowledge Distillation & CV \\
        100 & Robust Re-Identification by Multiple Views Knowledge Distillation & CV \\
        101 & Releasing Graph Neural Networks with Differential Privacy Guarantees & Graph \\
        102 & Multi-Level Branched Regularization for Federated Learning & Federated \\
        103 & Data-Free Knowledge Distillation for Deep Neural Networks & CV \\
        104 & MixKD: Towards Efficient Distillation of Large-scale Language Models & NLP \\
        105 & Enabling Multimodal Generation on CLIP via Vision-Language Knowledge Distillation & Multimodal \\
        106 & Learning Student Networks via Feature Embedding & CV \\
        107 & Distilling portable Generative Adversarial Networks for Image Translation & CV \\
        108 & AdaBERT: Task-Adaptive BERT Compression with Differentiable Neural Architecture Search & NLP \\
        109 & No One Left Behind: Inclusive Federated Learning over Heterogeneous Devices & Federated \\
        110 & Learning Compact Metrics for MT & NLP \\
        111 & SEED: Self-supervised Distillation For Visual Representation & CV \\
        112 & Cross-modal Knowledge Distillation for Vision-to-Sensor Action Recognition & CV \\
        113 & R2L: Distilling Neural Radiance Field to Neural Light Field for Efficient Novel View Synthesis & CV \\
        114 & Towards Efficient 3D Object Detection with Knowledge Distillation & CV \\
        115 & Knowledge distillation from multi-modal to mono-modal segmentation networks & CV \\
        116 & Distilling Knowledge from Deep Networks with Applications to Healthcare Domain & CV \\
        117 & Learning Perception-Aware Agile Flight in Cluttered Environments & RL/Policy \\
        118 & Adapt-and-Distill: Developing Small, Fast and Effective Pretrained Language Models for Domains & NLP \\
        119 & Knowledge Distillation for Small-footprint Highway Networks & Speech \\
\end{longtable}
}

\section{Recursion Diagnostics}
\label{app:diagnostics}

Distributional texture of the recursion, per model per round; the results
of the paper do not depend on any of it (Figure~\ref{fig:texture}). Gini over the round's shown papers
rises for every model, from 0.37--0.48 at round 1 to 0.54--0.64 at round
11, ordered like the map-side quantities of
Section~\ref{sec:recursion}; the top-1\% citation share rises in parallel
(roughly 4\% to 5--6\%). Openness to new work mirrors the funnel. The share of citations going to
the newest available cohort sits below that cohort's exposure share for every lifted model (5.8--8.9\% against 9.2\% at
round 11), while the two near-liftless models sit at or slightly above it
(9.6\% and 10.2\%); a lift toward the old necessarily starves the new, and
the models without the lift do not. Top-decile turnover between consecutive
rounds is high for all models (top-decile turnover 0.86--0.92 at round 11),
which is mechanical while the catalogue grows by 120 papers a round;
turnover is why the elite's identity churns even as its \emph{share}
concentrates.

\begin{figure}[ht]
    \centering
    \includegraphics[width=\textwidth]{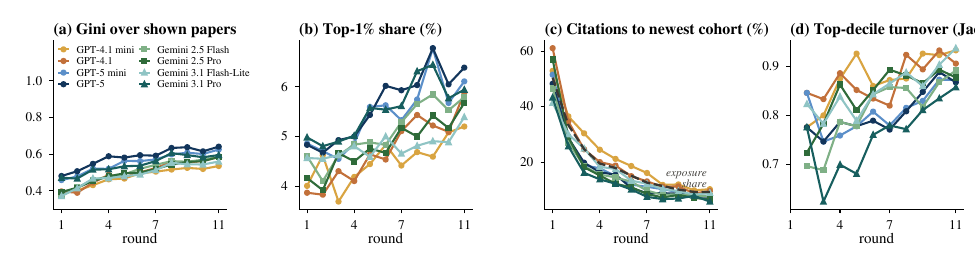}
    \caption{\textbf{Recursion texture} (8 models, rounds 1--11).
    (a)~Gini of citation counts over shown papers. (b)~Top-1\% citation
    share. (c)~Share of citations to the newest available cohort against
    its exposure share (dashed): lifted models under-cite the new.
    (d)~Top-decile turnover (Jaccard distance) between consecutive rounds.}
    \label{fig:texture}
\end{figure}

\end{document}